\documentclass[12pt]{article}

\usepackage[margin=1in]{geometry}
\usepackage{amsmath,amssymb,amsthm,mathtools}
\usepackage{bm}
\usepackage{enumitem}
\usepackage{booktabs}
\usepackage{xcolor}
\usepackage{hyperref}
\usepackage{natbib} 
\usepackage{makecell}
\usepackage{arydshln}
\usepackage{setspace} 
\hypersetup{colorlinks=true, linkcolor=red, citecolor=blue, urlcolor=blue}
\newtheorem{theorem}{Theorem}
\newtheorem{lemma}[theorem]{Lemma}
\newtheorem{proposition}[theorem]{Proposition}
\newtheorem{definition}[theorem]{Definition}

\newtheorem{condition}[theorem]{Condition}
\theoremstyle{definition}
\newtheorem{remark}[theorem]{Remark}
\newcommand{\E}{\mathbb{E}}
\newcommand{\Var}{\mathrm{Var}}

\renewcommand{\Pr}{\mathrm{Pr}}

\newcommand{\Xb}{\bar{X}}
\newcommand{\Yb}{\bar{Y}}
\newcommand{\Nb}{\bar{N}}
\newcommand{\Uh}{\hat{U}}
\newcommand{\Oh}{\hat{O}}
\DeclareMathOperator{\expit}{expit}
\DeclareMathOperator{\logit}{logit}
\newcommand{\bh}{\hat{\beta}}
\newcommand{\Sh}{\hat{\Sigma}}
\newcommand{\opn}[1]{o_p\big(#1\big)}
\newcommand{\Opn}[1]{O_p\big(#1\big)}

\providecommand{\keywords}[1]{\vspace{0.5em}\noindent\textbf{Keywords:} #1}

\title{Bridging Balancing Weights and Augmentation in Covariate-adjusted Analyses with Time-to-Event Endpoints: Theory and Practical Recommendations}

\author{
  Baoshan Zhang$^{1,}$\thanks{These authors contributed equally to this work.}
  \and Yi Chen$^{2,}$\footnotemark[\value{footnote}]
  \and Yu Du$^{2}$
  \and Tuo Wang$^{2,}$\thanks{Corresponding author: \texttt{tuo.wang@lilly.com}}
}
\date{\today}
\begin{document}
\maketitle
\begin{center}
$^{1}$ Department of Biostatistics and Bioinformatics, Duke University, Durham, NC\\
$^{2}$ Department of Global Statistical Sciences and Advanced Analytics, Eli Lilly and Company, Indianapolis, IN 
\end{center}

\begin{abstract}
Covariate adjustment improves the efficiency of treatment-effect analyses in randomized clinical trials, provided the adjustment targets the correct quantity. For time-to-event endpoints, two marginal targets are of primary interest: the log-rank test for the presence of a treatment effect and the marginal hazard ratio for its magnitude. Existing covariate adjustment approaches reach these targets by different ways. Augmentation adjusts the log-rank score by regressing derived outcomes on the baseline covariates within each arm. Weighting instead reweights the two arms to balance the covariates before the survival comparison is formed: inverse probability weighting does so through a fitted propensity model, while calibration weighting solves directly for weights that match covariate means. In this manuscript, we first develop balancing weighting for time-to-event endpoints, covering both calibration weights (stable balancing weights and entropy balancing) and propensity score weights, and prove that any balancing-regular weighting is first-order equivalent to the augmented log-rank score and to the root of the marginal Cox score. All three routes therefore deliver the same estimator to first order, and calibration reaches it without fitting any model. The weighted procedures thereby inherit the validity and guaranteed efficiency gain of the augmentation approach. In addition, we show that the efficiency gain grows with the prognostic strength of the adjustment covariates, while the practical caveat lies in variance estimation, for which we give recommendations to guard against finite-sample Type I error inflation. We further confirm our results through simulation studies and an analysis of the REWIND cardiovascular trial.
\end{abstract}


\newpage

\section{Introduction}
Randomized clinical trials are the standard for evaluating treatment
effects, and adjusting for prognostic baseline covariates improves
efficiency without introducing bias. The practice is endorsed by
regulators \citep{fda2023covariate} and supported by a large literature
\citep{ye2023toward, bannick2026covariate, shao2026should,ye2020robust,lu2008improving},
with well-established methods for continuous and binary endpoints
\citep{tsiatis2008covariate, lin2013agnostic, wang2019analysis, moore2009covariate}. Time-to-event endpoints are usually summarized by two targets: the log-rank test for whether the treatment has an effect and the marginal hazard ratio (HR) for how large that effect is. Both are marginal, population-level summaries and among the estimands a confirmatory analysis may target \citep{fda2023covariate, vanlancker2024covariate}.  This paper focuses on the marginal targets. Conditional analyses, by contrast, address a different estimand. They are equally legitimate but lie outside our scope. For the marginal targets, bringing covariates into a Cox proportional hazards model comes at a cost: the fitted HR is then a \emph{conditional} one, and because the HR is non-collapsible, it generally differs from the \emph{marginal} HR and depends on which covariates are entered.
The goal is therefore to use covariates to gain efficiency on the
log-rank test and the marginal HR without changing the
estimand or adding modeling assumptions \citep{gail1984biased}.
 
One line of work achieves this by \emph{augmentation} adjustment. \citet{ye2024covariate} augmented the log-rank score by linearizing it into per-subject derived outcomes and regressing these on the covariates within each arm, yielding a guaranteed efficiency gain that stays valid under simple and covariate-adaptive randomization; the same construction adjusts the marginal Cox HR, covering both marginal targets without changing the estimand. \citet{zhang2025unified} later gave a more general augmentation family, allowing flexible, possibly machine-learning, working models as the linear version is a special case, and practical guidance for oncology trials has begun to appear \citep{backenroth2026practical}.
 
A design-side alternative to augmentation is to reweight  the two arms so that they are balanced on the covariates. Such \emph{weighting estimators} come in two forms, distinguished by how the weights are constructed \citep{chattopadhyay2020balancing}. \emph{Propensity-score weights}, such as inverse probability of treatment weighting (IPW), are built from a fitted propensity model. Outside the time-to-event setting they reach the same efficiency as augmentation, improving precision for continuous and binary endpoints in the randomized trials \citep{williamson2014variance}, but the propensity model must still be specified and estimated. \emph{Calibration weights} fit no model at all: originating in observational studies, entropy balancing (EB) \citep{hainmueller2012entropy} and stable balancing weights (SBW) \citep{zubizarreta2015stable} solve directly for the weights closest to uniform that match pre-specified covariate moments exactly, giving exact finite-sample balance, and an analysis nearly as simple as the unadjusted one. For time-to-event endpoints, \citet{shao2026inverse} showed that IPW is asymptotically equivalent to the optimal linear augmentation for the marginal HR. For calibration weights, it remains open whether they can adjust both the log-rank test and the marginal HR, and whether they reduce to augmentation; more broadly, it is unclear whether the calibration and propensity-score families fall under a single equivalence.
 
In this paper, we show that three seemingly different routes to covariate adjustment ultimately lead to the same first-order correction. Augmentation uses linear regression on the derived outcome \citep{ye2024covariate}, whereas \citet{shao2026inverse} obtains the same efficiency gain through estimated propensity-score weighting. We complete this picture by showing that calibration weights, including SBW and EB, reproduce the same linear adjustment without fitting either an outcome or propensity-score model. More generally, any balancing-regular weighting scheme is first-order equivalent to augmented log-rank and root of marginal Cox scores. Thus, derived outcome augmentation adjustment, propensity-score weighting, and calibration weighting are different paths to the same destination. The resulting weighted procedures inherit the corresponding validity and efficiency properties of augmentation method. 

This article is organized as follows. In
Section~\ref{sec:background}, we present the trial setup, notation, and
marginal targets, and review the augmentation route. In
Section~\ref{sec:balancing}, we develop the balancing-weight route and
establish its equivalence to augmentation for both the log-rank test and
the marginal HR. In Section~\ref{sec:practical}, we give
practical recommendations on the efficiency gain and variance
estimation. In Sections~\ref{sec:simulation} and \ref{sec:realdata}, we
evaluate the proposed methods using simulation studies and the REWIND
trial. In Section~\ref{sec:discussion}, we conclude with a discussion.

\section{Background: Augmentation for Survival Endpoints}
\label{sec:background}

\subsection{Trial setup, notation, and marginal targets}
\label{sec:setup}

We consider a two-arm randomized clinical trial that compares an experimental treatment $(A=1)$ with a control treatment $(A=0)$ on a right-censored time-to-event endpoint. This work can also be extended to multiple arms and here we consider only two arms. For subject \(i=1,\ldots,n\), let \(A_i\in\{0,1\}\) denote the treatment assignment and let \(X_i\) denote a vector of observed baseline covariates used for covariate adjustment. Suppose the allocation ratio is fixed as $\pi$ in this randomized controlled setting. 

For subject $i$ under treatment $j=0,1$, let $T_{ij}^*$ and $C_{ij}$ denote the potential failure time and the potential censoring time. We assume that censoring is noninformative within each treatment group, that is, $T_A$ and $C_A$ are conditionally independent given $A$ \citep{lu2008improving, ye2024covariate}. Because follow-up is limited and dropout may occur, we observe the time $T_i=\min(T_{iA_i}^*,C_{iA_i})$ and the event indicator $\delta_i=\mathbf 1\{T_{iA_i}^*\le C_{iA_i}\}$. Define the potential at-risk and counting processes as
$
Y_{ij}(t)=\mathbf 1\{\min(T_{ij}^*,C_{ij})\ge t\},$ and $
N_{ij}(t)=\mathbf 1\{T_{ij}^*\le t,\ T_{ij}^*\le C_{ij}\}.
$ The observed at-risk is
$ Y_i(t)=A_iY_{i1}(t)+(1-A_i)Y_{i0}(t)$ and observed counting process is $
N_i(t)=A_iN_{i1}(t)+(1-A_i)N_{i0}(t). $
We write the arm-specific and pooled averages as 
$
\bar Y_j(t)=\frac1n\sum_{i:A_i=j}Y_i(t),$ $
\bar Y(t)=\bar Y_0(t)+\bar Y_1(t),$ and $
\bar N(t)=\frac1n\sum_{i=1}^nN_i(t).$
In addition, we assume throughout that $\inf_{t\le\tau}\E\{Y_i(t)\}>c>0$ for some positive constant \(c\).

In this manuscript, we consider two marginal targets. The first is the testing target of the log-rank statistic, the equality of the marginal failure-time distributions between the two groups and it is a non-parametric test whose type I error control does not rely on any modeling assumptions. Unadjusted log-rank test uses $\Gamma_\mathrm{unadj}=n^{1/2}\hat U_\mathrm{unadj}/\hat\sigma_\mathrm{unadj}$,
where
\[
\hat U_\mathrm{unadj}=\frac1n\sum_{i=1}^n\int_0^\tau\Big\{A_i-\frac{\bar Y_1(t)}{\bar Y(t)}\Big\}dN_i(t),\quad \text{with }
\hat\sigma_\mathrm{unadj}^2=\frac1n\sum_{i=1}^n\int_0^\tau\frac{\bar Y_1(t)\bar Y_0(t)}{\bar Y(t)^2}dN_i(t).\]
The second target is the marginal log HR $\theta$ in the unadjusted Cox model $\lambda_1(t)=\lambda_0(t)\exp{(\theta)}$, where $\lambda_1(t)$ and $\lambda_0(t)$ are the marginal hazard functions in two arms.
We treat this as a working model that yields a concise summary of the
treatment effect and do not assume proportional hazards; whether the summary is appropriate for a given trial is important in practice but outside our scope. We write $\theta_0$ for the corresponding estimand, the probability limit of the unadjusted Cox estimator, and defer the associated score to Section~\ref{sec:cox-extension}.

\subsection{Augmentation: the covariate-adjusted log-rank and Cox scores}
\label{sec:augmentation}
Augmentation achieves covariate adjustment on the outcome side: it removes from the score the part predictable from the baseline covariates, leaving the marginal target unchanged while delivering a guaranteed efficiency gain \citep{ye2024covariate}. We adopt it as the benchmark adjustment throughout, and show in Section~\ref{sec:lr-bridge} that balancing weights reproduce it to first order.

The unadjusted log-rank score $\hat U_\mathrm{unadj}$ admits the exact linearized representation via the derived outcome $\hat O_{ij}$ \citep{ye2024covariate, bannick2026covariate}: 
\[
\hat U_\mathrm{unadj}=\frac1n\sum_{i=1}^n\{A_i\hat O_{i1}-(1-A_i)\hat O_{i0}\}, \quad \text{where }
\hat O_{ij}
=\int_0^\tau\frac{\bar Y_{1-j}(t)}{\bar Y(t)}
\left\{dN_{ij}(t)-Y_{ij}(t)\frac{d\bar N(t)}{\bar Y(t)}\right\}.
\]
Although its explicit form is somewhat involved, $\hat O_{ij}$ has a simple interpretation: it is subject $i$'s influence-function contribution to the
log-rank score, measuring how much that subject's event and at-risk history deviate from the pooled experience.

Throughout, denote \(n_j=\sum_{i=1}^n\mathbf 1(A_i=j)\),
\(\Xb_j=n_j^{-1}\sum_{i:A_i=j}X_i\), and
\(\Xb=n^{-1}\sum_{i=1}^nX_i\). The representation above expresses $\hat U_\mathrm{unadj}$ as a difference in mean derived outcomes $\hat O_{ij}$ between the two arms, which makes covariate adjustment straightforward: treating $\hat O_{ij}$ as the outcome and applying the analysis of heterogeneous covariance (ANHECOVA) adjustment \citep{ye2020robust}. For arm $j$, let 
\[ 
\hat\beta_j
=\Big\{\sum_{i:A_i=j}(X_i-\Xb_j)(X_i-\Xb_j)^\top\Big\}^{-1}
\sum_{i:A_i=j}(X_i-\Xb_j)\hat O_{ij}. 
\]
be the within-arm ordinary-least-squares slope from regressing $\hat O_{ij}$ on $X_i$ with an intercept. The augmented score subtracts these within-arm fits,
\begin{equation}\label{eq:UCL}
\begin{aligned}
\hat U_\mathrm{aug}
&=\frac1n\sum_{i=1}^n
\Big\{A_i\big[\hat O_{i1}-(X_i-\Xb)^\top\hat\beta_1\big]
-(1-A_i)\big[\hat O_{i0}-(X_i-\Xb)^\top\hat\beta_0\big]\Big\}\\
&=\hat U_\mathrm{unadj}
-\frac1n\sum_{i=1}^n
\big\{A_i(X_i-\Xb)^\top\hat\beta_1-(1-A_i)(X_i-\Xb)^\top\hat\beta_0\big\}.
\end{aligned}
\end{equation}
The second line shows that $\hat U_\mathrm{aug}$ is the unadjusted score minus augmentation terms from both arms. Both are asymptotically mean zero under randomization. Thus, $\hat U_\mathrm{aug}$ and $\hat U_\mathrm{unadj}$ target the same marginal estimand.

These two scores thus differ only in efficiency. The augmented test statistic is  $\Gamma_\mathrm{aug}=\sqrt n\,\hat U_\mathrm{aug}/\hat\sigma_\mathrm{aug}$, with
$ \hat\sigma_\mathrm{aug}^2 =\hat\sigma_\mathrm{unadj}^2
-\pi(1-\pi)(\hat\beta_1+\hat\beta_0)^\top \hat\Sigma_X(\hat\beta_1+\hat\beta_0),$
where $\hat\Sigma_X$ is the sample covariance matrix of $X_i$. Under the stated randomization and censoring conditions, $\Gamma_\mathrm{aug}$ is valid and
has a guaranteed asymptotic efficiency gain over the unadjusted log-rank test
$\Gamma_\mathrm{unadj}$ \citep{ye2024covariate}.

The same augmentation  applies to estimation of the Cox marginal estimand $\theta_0$. Cox score for $\theta$ parallels the log-rank numerator with unadjusted root $\hat\theta_\mathrm{unadj}$. Augmenting the score by analogous within-arm slopes gives the augmented root $\hat\theta_\mathrm{aug}$, which targets the same marginal as $\hat\theta_\mathrm{unadj}$ while reducing its variance \citep{ye2024covariate}. We defer the construction and its variance estimation to Section~\ref{sec:cox-extension}.

\section{A Balancing-Weight Route and Its Equivalence to Augmentation}
\label{sec:balancing}
Section~\ref{sec:augmentation} adjusted for covariates by augmentation: regressing the derived outcomes $\hat{O}_{ij}$ on the baseline covariates $X$. This section develops a design-side route instead. Rather than touching the outcomes, we re-weight the two arms so that their baseline covariate distributions are balanced before the survival comparison is formed, and compute the log-rank score on the weighted sample. The two routes thus start from opposite ends: augmentation works on the outcome side and fits a regression, whereas weighting works on the design side and fits nothing to the outcome.  Our main result bridges them: any balancing-regular weighting reproduces the augmentation of Section~\ref{sec:augmentation} to first order, and therefore delivers the same adjusted log-rank test and, in its Cox HR extension, the same marginal HR estimator. Weighting consequently inherits the guaranteed efficiency gain of augmentation while fitting no outcome model, with the adjusted covariates pre-specified through balance constraints.

\subsection{Weighting: Calibration and Propensity Score}
\label{sec:balancing-weights} 
We attach a nonnegative weight $h_i$ to each subject $i$ so that, after weighting, the two arms' covariate distributions align with the pooled trial sample. Two families reach this goal by different routes. \emph{Calibration weights} (SBW, EB) impose covariate balance directly through moment constraints and fit no model, whereas \emph{estimated-propensity weights} (IPW) achieve balance indirectly by fitting a working model and weighting by its fitted values. We describe each in turn.

\paragraph{Calibration weights}
Calibration weights are a standard covariate-adjustment device in observational studies and impose the balance target directly \citep{zubizarreta2015stable, hainmueller2012entropy}. For each arm $j$, we choose weights $h_i$ that match the weighted covariate mean to the pooled covariate mean $\Xb$, through the first-moment exact calibration constraints
\begin{equation}\label{eq:hconstraints}
\frac1{n_j}\sum_{i:A_i=j}h_i=1,
\qquad
\frac1{n_j}\sum_{i:A_i=j}h_iX_i=\Xb,
\qquad
h_i\ge0.
\end{equation}
The constraints are imposed within each arm separately. The first keeps the total weight in an arm equal to its sample size; the second is a mean-balance constraint that removes the empirical mean imbalance in $X$ between that arm and the pooled trial sample.

The constraints in \eqref{eq:hconstraints} define the balance target; an objective function then selects, among all feasible weights $h_i$ meeting the target, the set closest to the uniform weights under a chosen criterion. Two choices are commonly used. SBW \citep{zubizarreta2015stable} minimizes the
weight variance, while EB \citep{hainmueller2012entropy} minimizes the negative entropy: 
\begin{equation}\label{eq:sbw_eb_objective}
\min_{\{h_i:A_i=j\}}\frac{1}{n_j}\sum_{i:A_i=j}(h_i-1)^2
\qquad\text{(SBW)},
\qquad
\min_{\{h_i:A_i=j\}}\frac{1}{n_j}\sum_{i:A_i=j}h_i\log h_i
\qquad\text{(EB)},
\end{equation}
each subject to \eqref{eq:hconstraints} for $j=0,1$.
Thus SBW and EB use the same calibration constraints but measure closeness to the uniform weights differently. Separating the balance target from the stability objective aids interpretation: the covariates to be balanced are pre-specified through \eqref{eq:hconstraints}, while the objective only prevents unnecessary weight variability.

\paragraph{Propensity weights.}
Estimated-propensity weighting provides a different construction. In a randomized trial the true propensity score is the known constant \(\pi\), so a working propensity model is correctly specified at \(\gamma_0=\mathbf 0\). Using that known constant alone would give constant weights and hence no covariate adjustment;  fitting a model in \(X\) lets the weights absorb chance treatment--covariate imbalance in the realized sample \citep{shao2026inverse,zeng2021propensity}. We use the centered logistic working model
\begin{equation}\label{eq:CentLog}
e_i(\gamma)
=
\expit\{\logit(\pi)+(X_i-\Xb)^\top\gamma\},
\qquad
\hat e_i=e_i(\hat\gamma),
\end{equation}
where the intercept is fixed at \(\logit(\pi)\) and only the slope is estimated. Equivalently, \(\hat\gamma\) solves the centered score equation $\sum_{i=1}^n 
\{A_i-e_i(\hat\gamma)\}(X_i-\Xb)=0.$ The stabilized IPW \citep{shao2026inverse} weights for subject $i$ are 
\begin{equation}\label{eq:IPW}
    h_i=A_i\frac{\pi}{\hat e_i}+(1-A_i)\frac{1-\pi}{1-\hat e_i}.
\end{equation}
Unlike calibration weights, the resulting $h_i$ satisfy \eqref{eq:hconstraints} only approximately: mean balance is attained asymptotically rather than exactly in finite samples.

\subsection{Weighted Log-rank Score}
\label{sec:weighted-scores}
Given any non-negative weights $h_i$, we form the log-rank score on the reweighted sample. Following \cite{xie2005adjusted}, define the arm-specific weighted risk-set and
counting-process averages
\[
\bar Y_j^{(h)}(t)
=
\frac1n\sum_{i:A_i=j}h_iY_i(t),\quad \bar N_j^{(h)}(t)
=
\frac1n\sum_{i:A_i=j}h_iN_i(t),
\]
with pooled average
$
\Yb^{(h)}(t)=\Yb_0^{(h)}(t)+\Yb_1^{(h)}(t),$ and $
\bar N^{(h)}(t)=\bar N_0^{(h)}(t)+\bar N_1^{(h)}(t);$ we assume \(\Yb^{(h)}(t)>0\) over the time interval over which the integrals are evaluated. The weighted log-rank score is 
\begin{equation}\label{eq:Uhdef}
\hat U_h
=
\frac1n\sum_{i=1}^n
h_i
\int_0^\tau
\Big\{
A_i-\frac{\bar Y_1^{(h)}(t)}{\bar Y^{(h)}(t)}
\Big\}
\,dN_i(t).
\end{equation}
which uses the same weights in the event contribution and in the risk-set composition. This is the direct weighted analogue of the unadjusted log-rank score $\hat{U}_\mathrm{unadj}$, with covariate balance imposed before the survival comparison is formed. 

The weighted score $\hat U_h$ also admits the same derived outcome linear representation as in $\hat{U}_\mathrm{unadj}$. Define the weighted derived outcomes
\[ \hat O_{ij}^{(h)}
  =\int_0^\tau
  \frac{\bar Y_{1-j}^{(h)}(t)}{\bar Y^{(h)}(t)}
  \Big\{
    dN_{ij}(t)-Y_{ij}(t)\frac{d\bar N^{(h)}(t)}{\bar Y^{(h)}(t)}
  \Big\}.
\]
Then, for any nonnegative weights $h_i$, we can show that \begin{equation}\label{eq:UhExactID}
      \hat U_h
  \equiv
  \frac{1}{n}\sum_{i=1}^n h_i
  \left\{A_i\hat O_{i1}^{(h)}-(1-A_i)\hat O_{i0}^{(h)}\right\}.    
\end{equation}
This identity holds exactly in finite samples (See proof in Supplementary Material) and it is the starting point for the equivalence argument.

\subsection{Balancing-regular weights} \label{sec:main-results}

We consider simple randomization and stratified designs, including both stratified simple randomization and stratified permuted-block randomization. All of these satisfy the following condition, under which our results are stated; the Supplementary Material verifies that each design meets it.

\begin{condition}[Randomization and covariate regularity]
\label{cond:car-design}
The randomization scheme has a fixed treatment allocation rate $\pi\in(0,1)$. If stratified randomization is used, the stratification variable
$Z_i$ is discrete with finitely many joint levels $\mathcal Z$, treatment
assignment is conditionally independent of the potential outcomes and baseline
covariates given $(Z_1,\ldots,Z_n)$, $\E(A_i\mid Z_1,\ldots,Z_n)=\pi$, and
$n_{z1}/n_z-\pi=\Opn{n^{-1/2}}$ for every $z\in\mathcal Z$. The covariate vector $X_i$ satisfies $\E\|X_i\|^2<\infty$ and $\Sigma_X=\Var(X_i)\succ0$, and a
full-rank coding of $Z_i$ is included in $X_i$ when stratified randomization is used.
\end{condition}

The equivalence we establish does not require exact finite-sample calibration and does not depend on the particular objective used to construct the weights. It depends only on the first-order shape of the weight perturbations: each $h_i-1$ must behave, to leading order, like the least-squares projection of the covariate imbalance. We make this precise below. Let $\hat\Sigma_{X|j}=n_j^{-1}\sum_{i:A_i=j}(X_i-\Xb_j)(X_i-\Xb_j)^\top$ be the arm-specific sample covariance matrix of $X_i$.

\begin{definition}[First-order balancing regularity]
\label{def:first-order-balancing}
Let \(h_i=1+\xi_i\) be arm-specific nonnegative weights.
We say that the weights are \emph{first-order balancing-regular}
if, for each arm \(j=0,1\),
\begin{equation}\label{eq:generic_cal_expansion_car}
\xi_i
=
(X_i-\Xb_j)^\top
\hat\Sigma_{X|j}^{-1}
(\Xb-\Xb_j)
+\rho_i,
\qquad i:A_i=j,
\end{equation}
where $\frac1n\sum_{i:A_i=j}\rho_i^2
=
o_p(n^{-1}),$
and also satisfy the arm-specific second-moment stability condition
$ \frac1n\sum_{i:A_i=j}\xi_i^2 = 
O_p(n^{-1}). $ Weights satisfying these conditions are called
\emph{balancing-regular}.
\end{definition}

The leading term in \eqref{eq:generic_cal_expansion_car} is the linear tilt that balances the arm-specific covariate mean $\Xb_j$ to the pooled mean $\Xb$, and the remainder $\rho_i$ absorbs any nonlinearity or objective-specific
deviation that is negligible at the $n^{-1/2}$ scale. Different weighting schemes
reach this same leading term by different routes; the next proposition verifies
the condition for the schemes considered here.

\begin{proposition}[Examples of balancing-regular weights]
\label{prop:weights-car}
Under Condition~\ref{cond:car-design}, the following weighting schemes are
balancing-regular.

\begin{enumerate}[nosep]
\item The SBW solution \eqref{eq:sbw_eb_objective} is balancing-regular with
\(\rho_i\equiv0\).

\item The EB solution \eqref{eq:sbw_eb_objective} is balancing-regular whenever the
EB problem is feasible and its dual multiplier satisfies the local regularity condition.

\item The IPW \eqref{eq:IPW} is balancing-regular when \(X_i\) has bounded support and the centered logistic score equation \eqref{eq:CentLog} admits a regular local solution.
\end{enumerate}
\end{proposition}
Thus calibration weights and estimated-propensity weights fall
inside the same first-order class. SBW attains the projection exactly, EB attains
it up to a negligible nonlinear-tilting remainder, and stabilized IPW attains it
through the first-order expansion of the fitted propensity score. The proofs are
given in the Supplementary Material.

\subsection{Balancing-regular weights reproduce the augmented log-rank score}
\label{sec:lr-bridge}
This subsection shows that any balancing-regular weighting scheme reproduces the augmented log-rank score to first order. The result
depends only on the first-order behavior of the weights, rather than
on how they are constructed---through a quadratic stability objective,
entropy tilting, or an estimated propensity-score equation---or on
whether covariate balance holds exactly or only to the order required
by Definition~\ref{def:first-order-balancing}. In this sense, the equivalence is \emph{objective-free}: it covers SBW, EB, and stabilized IPW as special cases, and the weighting route inherits the validity and guaranteed efficiency gain of augmentation without ever fitting an outcome model.

\begin{theorem}[First-order equivalence to the augmented log-rank score]
\label{the:wLR_bridge}
Under the randomization conditions,
suppose that the weights \(h_i\) are balancing-regular in the sense of
Definition~\ref{def:first-order-balancing}. Then
\begin{equation}\label{eq:Uh_unweighted_DO_rep}
\hat U_h
=
\frac1n\sum_{i=1}^n
h_i
\bigl\{
A_i\hat O_{i1}
-
(1-A_i)\hat O_{i0}
\bigr\}
+
\opn{n^{-1/2}}.
\end{equation}
Equivalently, writing \(h_i=1+\xi_i\),
\begin{equation}\label{eq:Uh_weight_perturbation}
\hat U_h-\hat U_\mathrm{unadj}
=
\frac1n\sum_{i=1}^n
\bigl\{
A_i\xi_i\hat O_{i1}
-
(1-A_i)\xi_i\hat O_{i0}
\bigr\}
+
\opn{n^{-1/2}}.
\end{equation}
Consequently,
\begin{equation}\label{eq:wLR_CL_equiv}
\hat U_h
=
\hat U_{\mathrm{aug}}
+
\opn{n^{-1/2}}.
\end{equation}
\end{theorem}

The key step in Theorem~\ref{the:wLR_bridge}  is the replacement of
weighted by unweighted derived outcomes. The exact finite-sample identity
\eqref{eq:UhExactID} represents $\hat U_h$  as a weighted
contrast of \(\hat O_{ij}^{(h)}\), which themselves depend on the weights through the perturbed risk sets. Equation~\eqref{eq:Uh_unweighted_DO_rep} decouples the derived outcomes from the weights: $\hat O_{ij}^{(h)}$ may be replaced by $\hat O_{ij}$ at a cost of only $\opn{n^{-1/2}}$. Once this replacement is made, the weights enter the score only through their perturbations \(\xi_i=h_i-1\), and \eqref{eq:Uh_weight_perturbation} follows from the exact unadjusted identity. This replacement is also why the equivalence is first-order rather than exact: unlike a linear estimator, the weighted log-rank score is nonlinear in the weights because they enter both the event contributions and the risk-set ratios \citep{brunssmith2025augmented}. The proof is given in the Supplementary Material.

The weighting and augmentation constructions approach adjustment from opposite directions. The augmented score \(\hat U_\mathrm{aug}\) works on the outcome side: it regresses \(\hat O_{ij}\) on \(X_i\) within each arm and subtracts the covariate-predictable component of \(\hat U_\mathrm{unadj}\). Weighting works on the design side: it modifies the arm-specific empirical distribution of \(X_i\) so that the arm-to-pooled covariate imbalance is corrected without using the outcomes.

Substituting the balancing-regular expansion \eqref{eq:generic_cal_expansion_car} into \eqref{eq:Uh_weight_perturbation} shows where the two routes meet. For arm \(j\),
\[
\begin{aligned}
\frac1n\sum_{i:A_i=j}\xi_i\hat O_{ij}
&=
\frac{n_j}{n}
(\Xb-\Xb_j)^\top
\hat\Sigma_{X|j}^{-1}
\underbrace{
\frac1{n_j}\sum_{i:A_i=j}
(X_i-\Xb_j)\hat O_{ij}
}_{\text{within-arm sample covariance}}
+
\opn{n^{-1/2}} \\
&=
\frac{n_j}{n}
(\Xb-\Xb_j)^\top
\hat\beta_j
+
\opn{n^{-1/2}}
\end{aligned}
\]
Thus, although the weights are built without reference to the outcomes, averaging their first-order covariate projection against
\(\hat O_{ij}\) reconstructs the within-arm covariance between the
covariates and the derived outcome, and hence the regression slope
\(\hat\beta_j\). The resulting correction is the product of two quantities arising from different sources: balancing supplies the arm-to-pooled covariate imbalance, the derived outcomes supply the prognostic slope. Combining the two arms and applying the centering identities yields exactly the augmentation correction, up to \(\opn{n^{-1/2}}\).

The practical implication is immediate. Let
\(\hat\sigma_{\mathrm{aug}}^2\) be a consistent estimator of the
asymptotic variance of \(\sqrt n\,\hat U_{\mathrm{aug}}\), and define
\[
\Gamma_h
=\sqrt n
\frac{\hat U_h}{\hat\sigma_{\mathrm{aug}}}.
\]
Theorem~\ref{the:wLR_bridge} gives
\[
\Gamma_h-\Gamma_{\mathrm{aug}}
=
\frac{
\sqrt n(\hat U_h-\hat U_{\mathrm{aug}})
}{
\hat\sigma_{\mathrm{aug}}
}
=
\opn{1},
\]
and Slutsky's theorem implies that the two statistics have the same
first-order limiting distribution. This conclusion applies to SBW, EB, stabilized IPW, and any other balancing-regular weighting scheme. The weighted log-rank test is therefore valid under the same randomization and censoring conditions and inherits the guaranteed asymptotic efficiency gain of the augmented test over the unadjusted one, while fitting no outcome model and allowing the adjusted covariates to be pre-specified through the balance constraints. Variance estimation choices are discussed in Section~\ref{sec:VarEst}.

\subsection{Extension to the Marginal Cox Score}
\label{sec:cox-extension}

Beyond testing the null of no treatment effect, it is often useful to estimate a summary measure of the treatment effect and construct a confidence interval \citep{zhang2015robust, parast2014landmark}. We use the working Cox model of Section~\ref{sec:setup}, with treatment indicator as the only covariate, 
\[
\lambda_1(t)
=
\lambda_0(t)\exp(\theta).
\]
This model imposes proportional hazards, which hold automatically under the null $\lambda_1(t)=\lambda_0(t)$ but need not hold otherwise. As in \citet{shao2026inverse}, we regard it as a working model that produces a succinct summary of the treatment effect without assuming proportional hazards, and take as our estimand $\theta_0$, the probability limit of the unadjusted
partial-likelihood estimator $\hat\theta_\mathrm{unadj}$.   Whether or not the model is correct, $\theta_0$ is the unique root of the asymptotic score equation and $\sqrt n\,(\hat\theta_\mathrm{unadj}-\theta_0)$ is asymptotically normal, with variance consistently estimated by the sandwich estimator \citep{lin1989robust}. In either case $\theta_0$ is the estimand the unadjusted analysis already targets, so the equivalence below does not depend on whether proportional hazards holds. Whether $\theta_0$ is an appropriate effect measure for a given trial is an important practical question, but not our focus.

For a working value $\vartheta$ of the Cox parameter, the unadjusted Cox partial-likelihood score is
\begin{equation}\label{eq:unadj_Ucox}
\Uh_\mathrm{unadj}(\vartheta)
=
\frac1n\sum_{i=1}^n
\int_0^\tau
\left\{
A_i-
\frac{e^\vartheta\Yb_1(t)}
{e^\vartheta\Yb_1(t)+\Yb_0(t)}
\right\}
dN_i(t).
\end{equation}
The maximum partial-likelihood estimator $\hat\theta_\mathrm{unadj}$ solves $\Uh_\mathrm{unadj}(\vartheta)=0$; the adjusted roots defined below target the same estimand $\theta_0$, so covariate adjustment improves precision without changing the estimand.

Following the idea of \citet{ye2024covariate}, the augmented Cox score subtracts from $\Uh_\mathrm{unadj}(\vartheta)$ the within-arm slopes frozen at $\hat\theta_\mathrm{unadj}$, 
\[
\hat U_\mathrm{aug}(\vartheta)
=\Uh_\mathrm{unadj}(\vartheta)
-\frac1n\sum_{i=1}^n
\Bigl\{
A_i(X_i-\Xb)^\top\hat\beta_1(\hat\theta_\mathrm{unadj})
-(1-A_i)(X_i-\Xb)^\top\hat\beta_0(\hat\theta_\mathrm{unadj})
\Bigr\},
\]
where, for $j=0,1$, $\hat\beta_j(\vartheta)$ is the within-arm least-squares
slope of $\hat O_{ij}(\vartheta)$ on $X_i-\Xb_j$, with the $\vartheta$-indexed
derived outcome
\[
\hat O_{ij}(\vartheta)
=\int_0^\tau
\frac{\{e^\vartheta\bar Y_1(t)\}^{1-j}\{\bar Y_0(t)\}^{j}}{e^\vartheta\Yb_1(t)+\Yb_0(t)}
\Bigl\{dN_{ij}(t)-Y_{ij}(t)\,\frac{e^\vartheta\,d\bar N(t)}{e^\vartheta\Yb_1(t)+\Yb_0(t)}\Bigr\}.
\]
These reduce to $\hat O_{ij}$ and $\hat\beta_j$ at $\vartheta=0$. The augmented root $\hat\theta_\mathrm{aug}$ solves $\hat U_\mathrm{aug}(\vartheta)=0$. Freezing the slopes at $\hat\theta_\mathrm{unadj}$ rather than at the running value $\vartheta$ makes $\hat U_\mathrm{aug}$ a one-step adjustment, and \citet{ye2024covariate} show that $\hat\theta_\mathrm{aug}$ has smaller variance than $\hat\theta_\mathrm{unadj}$.

Following the weighted Cox partial-likelihood score
\citep{binder1992fitting, lin2000fitting}, we attach a subject-level weight $h_i$ to both the event contribution and the risk set,
\begin{equation}\label{eq:Uhtheta}
\hat U_h(\vartheta)
=\frac1n\sum_{i=1}^n h_i\int_0^\tau
\left\{A_i-\frac{e^\vartheta\bar Y_1^{(h)}(t)}
{e^\vartheta\bar Y_1^{(h)}(t)+\bar Y_0^{(h)}(t)}\right\}dN_i(t),
\end{equation}
with weighted root $\hat\theta_h$ solving $\hat U_h(\vartheta)=0$. 

\begin{theorem}[First-order equivalence to the augmented Cox root]
\label{thm:cox-equiv}
Assume Condition~\ref{cond:car-design}, that the weights are
balancing-regular in the sense of
Definition~\ref{def:first-order-balancing}, and the Cox regularity
conditions stated in the Supplementary Material. Then
\[
\hat\theta_h
=
\hat\theta_\mathrm{aug}
+
\opn{n^{-1/2}}.
\]
Under the method-specific conditions given earlier, this conclusion
applies to SBW, EB, stabilized IPW and all other balancing-regular weights.
\end{theorem}

The weighted and augmented Cox roots therefore have the same first-order limiting distribution, and inherit the efficiency and variance results established for \(\hat\theta_\mathrm{aug}\) by \citet{ye2024covariate}. No separate variance theory is needed for the weighted route. The proof of Theorem~\ref{thm:cox-equiv} is given in the Supplementary Material.

\section{Practical Recommendations}\label{sec:practical}
Theorems~\ref{the:wLR_bridge} and~\ref{thm:cox-equiv} remove efficiency as a
criterion for choosing among the adjustment methods: to first order they deliver
the same test and the same estimator, so the choice among them can rest on
convenience. Two decisions do carry practical consequences. The first is which
covariates to adjust for: the weights are constructed without reference to the
outcomes, yet the gain comes entirely from the covariate--outcome association,
and Section~\ref{Sec:ProgEG} makes that dependence explicit. The second is how to
estimate the variance: the methods share a first-order limiting distribution, but
differ in which estimators are available and in how those estimators behave in
finite samples, which is where the only real failure mode appears.
Section~\ref{sec:VarEst} takes this up.
 
\subsection{Prognostic Covariate and Efficiency Gain}
\label{Sec:ProgEG}
The efficiency gain has a simple interpretation through derived outcome representation of the log-rank score. Suppressing the subject index, the
augmented log-rank variance satisfies
\[
\sigma_\mathrm{aug}^2
=
\sigma_\mathrm{unadj}^2
-
\pi(1-\pi)(\beta_1+\beta_0)^\top\Sigma_X(\beta_1+\beta_0), \quad
\beta_j=\Sigma_X^{-1}\mathrm{Cov}(X,O_j).
\]
The variance removed by adjustment is therefore governed by how strongly the baseline
covariates are linearly associated with the arm-specific derived outcomes
\(O_j\); equivalently, adjustment is most useful when the covariates explain the
patient-level log-rank contributions.  The matrix \(\Sigma_X\) mainly
standardizes this association by accounting for the scale and correlation of the
covariates.

To connect this abstract slope $\beta_j$ to a clinically interpretable notion of
prognostic value, consider the working model
\(\lambda_j(t\mid X)=\lambda_{0j}(t)\exp(\eta^\top X)\) as an illustrative example, where \(\eta\) is a
covariate log-hazard coefficient and is distinct from the marginal treatment
effect.  A local expansion around \(\eta=0\), detailed in the Supplement, gives
\(\beta_j=c_j\eta+O(\|\eta\|^2)\) with constant \(c_j>0\) under the local null baseline.
Consequently,
\[
\sigma_\mathrm{unadj}^2-\sigma_\mathrm{aug}^2
=
\pi(1-\pi)(c_1+c_0)^2\eta^\top\Sigma_X\eta
+
O(\|\eta\|^3).
\]
This calculation shows why more prognostic covariates lead to larger efficiency
gains: to leading order, the gain increases with the quadratic prognostic
strength \(\eta^\top\Sigma_X\eta\).  

\subsection{Variance estimation for the weighted analyses}
\label{sec:VarEst}

The first-order equivalence transfers variance theory for the augmented estimator \cite{ye2024covariate} to the weighted estimators, but in practice one must still choose a variance estimator that is both \emph{available} for a given method and \emph{valid} in the relevant sample size. 

Ready-made variance estimators differ across methods. For augmentation, an analytic standard error is returned by \texttt{RobinCar} \citep{ye2024covariate}, for both the log-rank test and the marginal HR. For IPW and EB, an asymptotic stacked M-estimation standard error is available in the \texttt{WeightIt} R package \citep{WeightIt,shu2021variance}; this estimator stacks the estimating equations for the weights and the outcome model, thereby accounting for the uncertainty in the estimated weights. For SBW, no packaged analytic variance is currently available, so a resampling estimator is the practical default. 

Analytic and M-estimation standard errors are asymptotically correct but can underestimate the true variability when the sample is small or the number of adjusted covariates is large \cite{austin2013performance} and \citet{shao2026inverse} report the same tendency for the analytic covariate-adjusted variance. Because a test statistic divides by the estimated standard error, such underestimation inflates the Type~I error rate; our own
simulations in Section~\ref{sec:Sim_Infer} show this for augmentation, EB and
IPW alike. The bootstrap has been found to estimate standard errors and confidence intervals well after weighting, sometimes even outperforming the asymptotically correct estimator in smaller samples \citep{austin2022bootstrap,ZhangVarEst2026}, at the cost of refitting the procedure in each resample. When Type~I error  control is the priority and computation permits, the bootstrap is therefore the safer choice; when a validated analytic or M-estimation form is available and the sample is not small relative to the number of covariates, it offers the same accuracy at far lower cost.

\section{Simulation Study}
\label{sec:simulation}
We conducted three simulation studies to evaluate the finite-sample behavior of the weighted estimators developed in
Section~\ref{sec:balancing}: \emph{equivalence}, \emph{inference validity} and \emph{prognostic effect}.  The equivalence study examined whether the weighting implementations reproduce the point estimates and efficiency gains of augmentation for the log-rank score and the marginal Cox HR. The inference validity study evaluated Type~I error and standard-error calibration for the marginal HR. The \emph{prognostic effect} study examined how the efficiency gain varies with the prognostic strength of the baseline covariates, as characterized in Section~\ref{Sec:ProgEG}. The unadjusted method served as the common benchmark throughout.

\subsection{Data Generation Process}
The baseline covariate vector $X_i$ was motivated by components of the Framingham risk score: $L_{\mathrm{AGE},i}$, $L_{\mathrm{HDL},i}$ and $L_{\mathrm{SBP},i}$, the logarithms of age, HDL cholesterol and systolic blood pressure, together with diabetes and smoking indicators.  The three log-transformed continuous components were drawn from a multivariate normal distribution with correlation structure calibrated to the Framingham population \citep{d2008general}, and the two binary components from independent Bernoulli distributions. A nonlinear oracle prognostic score $Z_i$ was constructed from all components of $X_i$: $Z_i=1-S_0^{\exp(\alpha^\top X_i)}$ with $S_0=0.9$, the same form used by Framingham-type risk scores. Full parameter
values are given in the Supplementary Material.

Under simple randomization, treatment was assigned as
$A_i\sim\operatorname{Bernoulli}(0.5)$ and event times followed an exponential proportional-hazards model, $\lambda(t\mid A_i,Z_i)=0.05*\exp(\beta_A A_i+\beta_Z Z_i).$
Administrative censoring was generated independently as
$C\sim\operatorname{Uniform}(1,4)$, representing three years of uniform accrual and one additional year of follow-up. 

\subsection{Simulation Study: Empirical equivalence}

The empirical equivalence study used $\beta_A=\log(0.7)$, $\beta_Z=9$, and 10,000 Monte Carlo (MC) replicates per sample size. Table~\ref{tab:s1-equivalence-total} reports $n\in\{400,800\}$; the paired rate diagnostic additionally used $n\in\{1200,1600,2000\}$. 

Seven scenarios were considered: (A) \emph{Oracle score}: adjustment using $Z$;  (B) \emph{All Baseline}: adjustment using all five components of $X_i$; (C) \emph{Over-adjusted}: adjustment using all five baseline covariates plus three independent $N(0,1)$ noise variables; (D) \emph{Omit Key}: omission of the strongly prognostic $L_{\mathrm{SBP},i}$ component, with one independent $N(0,1)$ noise variable added; (E) \emph{Wrong Form}: replacement of $L_{\mathrm{AGE},i}$ by its square; and two stratified designs, (F) \emph{Stratified Simple}; (G) \emph{Permuted Block} (size four),  both stratified on $S_i=\mathbf 1(Z_i<0.15)$, with
implementation details in the Supplementary Material.

This simulation study deliberately evaluated point estimators only: no analytic or bootstrap standard errors were computed. For each method, we report the MC mean, empirical standard deviation (ESD), and the relative efficiency (RE)
\[
\operatorname{RE}
=
\operatorname{ESD}^2(\widehat\psi_{\mathrm{unadj}})
\big/
\operatorname{ESD}^2(\widehat\psi_{\mathrm{adj}}),
\]
where $\widehat\psi$ denotes the log-rank score for the log-rank and the estimated marginal log HR for the Cox, and the subscript $\mathrm{adj}$ refers generically to any of the four adjusted estimators.

\begin{table}[htbp]
\centering
\caption{Empirical equivalence results based on 10,000 MC replicates. 
}
\label{tab:s1-equivalence-total}
\scriptsize
\setlength{\tabcolsep}{2pt}
\resizebox{\textwidth}{!}{%
\begin{tabular}{ll rrc rrc rrc rrc}
\toprule
 & & \multicolumn{6}{c}{$n=400$} & \multicolumn{6}{c}{$n=800$} \\
\cmidrule(lr){3-8} \cmidrule(lr){9-14}
 & & \multicolumn{3}{c}{Log-rank test} & \multicolumn{3}{c}{Marginal HR} & \multicolumn{3}{c}{Log-rank test} & \multicolumn{3}{c}{Marginal HR} \\
\cmidrule(lr){3-5} \cmidrule(lr){6-8} \cmidrule(lr){9-11} \cmidrule(lr){12-14}
Scenario & Method & Mean & ESD & RE & Mean & ESD & RE & Mean & ESD & RE & Mean & ESD & RE \\
\midrule
Oracle Score & Unadj. & -0.021 & 0.014 &  --- & -0.272 & 0.178 &  ---& -0.021 & 0.010 & --- & -0.271 & 0.128 & --- \\
 & Aug. & -0.022 & 0.012 & 1.360 & -0.272 & 0.152 & 1.361 & -0.022 & 0.009 & 1.343 & -0.272 & 0.110 & 1.347 \\
 & SBW & -0.022 & 0.012 & 1.360 & -0.272 & 0.152 & 1.363 & -0.022 & 0.009 & 1.343 & -0.272 & 0.110 & 1.348 \\
 & EB & -0.021 & 0.012 & 1.361 & -0.272 & 0.152 & 1.363 & -0.022 & 0.009 & 1.344 & -0.272 & 0.110 & 1.349 \\
 & IPW & -0.022 & 0.012 & 1.361 & -0.272 & 0.152 & 1.368 & -0.022 & 0.009 & 1.344 & -0.272 & 0.110 & 1.352 \\
\addlinespace[2pt]
All Baseline & Unadj. & -0.021 & 0.014 & --- & -0.272 & 0.178 &  ---& -0.021 & 0.010 & --- & -0.271 & 0.128 &  ---\\
 & Aug. & -0.021 & 0.012 & 1.258 & -0.272 & 0.159 & 1.252 & -0.022 & 0.009 & 1.263 & -0.272 & 0.114 & 1.261 \\
 & SBW & -0.021 & 0.012 & 1.258 & -0.272 & 0.159 & 1.253 & -0.022 & 0.009 & 1.263 & -0.272 & 0.114 & 1.262 \\
 & EB & -0.021 & 0.012 & 1.258 & -0.272 & 0.159 & 1.255 & -0.022 & 0.009 & 1.264 & -0.272 & 0.114 & 1.263 \\
 & IPW & -0.022 & 0.012 & 1.250 & -0.272 & 0.159 & 1.254 & -0.022 & 0.009 & 1.260 & -0.272 & 0.114 & 1.263 \\
\addlinespace[2pt]
Over-adjusted & Unadj. & -0.021 & 0.014 &  ---& -0.272 & 0.178 & --- & -0.021 & 0.010 & --- & -0.271 & 0.128 &  ---\\
 & Aug. & -0.021 & 0.012 & 1.247 & -0.272 & 0.160 & 1.242 & -0.022 & 0.009 & 1.261 & -0.272 & 0.114 & 1.259 \\
 & SBW & -0.021 & 0.012 & 1.248 & -0.272 & 0.159 & 1.244 & -0.022 & 0.009 & 1.261 & -0.272 & 0.114 & 1.260 \\
 & EB & -0.021 & 0.012 & 1.247 & -0.272 & 0.159 & 1.244 & -0.022 & 0.009 & 1.261 & -0.272 & 0.114 & 1.260 \\
 & IPW & -0.022 & 0.013 & 1.238 & -0.272 & 0.160 & 1.242 & -0.022 & 0.009 & 1.258 & -0.272 & 0.114 & 1.261 \\
\addlinespace[2pt]
Omit Key & Unadj. & -0.021 & 0.014 & --- & -0.272 & 0.178 & --- & -0.021 & 0.010 & --- & -0.271 & 0.128 &  ---\\
 & Aug. & -0.021 & 0.013 & 1.231 & -0.272 & 0.161 & 1.225 & -0.022 & 0.009 & 1.235 & -0.272 & 0.115 & 1.234 \\
 & SBW & -0.021 & 0.013 & 1.231 & -0.272 & 0.160 & 1.227 & -0.022 & 0.009 & 1.235 & -0.272 & 0.115 & 1.235 \\
 & EB & -0.021 & 0.013 & 1.230 & -0.272 & 0.160 & 1.228 & -0.022 & 0.009 & 1.236 & -0.272 & 0.115 & 1.236 \\
 & IPW & -0.022 & 0.013 & 1.222 & -0.272 & 0.161 & 1.226 & -0.022 & 0.009 & 1.232 & -0.272 & 0.115 & 1.236 \\
\addlinespace[2pt]
Wrong Form & Unadj. & -0.021 & 0.014 & --- & -0.272 & 0.178 &  ---& -0.021 & 0.010 & --- & -0.271 & 0.128 &  ---\\
 & Aug. & -0.021 & 0.012 & 1.268 & -0.272 & 0.158 & 1.262 & -0.022 & 0.009 & 1.273 & -0.272 & 0.113 & 1.271 \\
 & SBW & -0.021 & 0.012 & 1.268 & -0.272 & 0.158 & 1.264 & -0.022 & 0.009 & 1.273 & -0.272 & 0.113 & 1.272 \\
 & EB & -0.021 & 0.012 & 1.269 & -0.272 & 0.158 & 1.265 & -0.022 & 0.009 & 1.273 & -0.272 & 0.113 & 1.273 \\
 & IPW & -0.022 & 0.012 & 1.261 & -0.272 & 0.158 & 1.264 & -0.022 & 0.009 & 1.270 & -0.272 & 0.113 & 1.273 \\
\addlinespace[2pt]
Stratified Simple & Unadj. & -0.021 & 0.014 & --- & -0.272 & 0.179 & --- & -0.022 & 0.010 &  ---& -0.273 & 0.128 & --- \\
 & Aug. & -0.021 & 0.012 & 1.341 & -0.271 & 0.155 & 1.338 & -0.022 & 0.009 & 1.346 & -0.274 & 0.110 & 1.355 \\
 & SBW & -0.021 & 0.012 & 1.341 & -0.270 & 0.155 & 1.340 & -0.022 & 0.009 & 1.346 & -0.274 & 0.110 & 1.356 \\
 & EB & -0.021 & 0.012 & 1.342 & -0.270 & 0.155 & 1.341 & -0.022 & 0.009 & 1.346 & -0.274 & 0.110 & 1.356 \\
 & IPW & -0.021 & 0.012 & 1.344 & -0.270 & 0.154 & 1.348 & -0.022 & 0.009 & 1.347 & -0.274 & 0.109 & 1.359 \\
\addlinespace[2pt]
Permuted Block & Unadj. & -0.022 & 0.013 & --- & -0.272 & 0.162 &  --- & -0.022 & 0.009 & --- & -0.271 & 0.113 & --- \\
 & Aug. & -0.022 & 0.012 & 1.118 & -0.273 & 0.153 & 1.118 & -0.022 & 0.008 & 1.122 & -0.271 & 0.107 & 1.120 \\
 & SBW & -0.022 & 0.012 & 1.118 & -0.273 & 0.153 & 1.119 & -0.022 & 0.008 & 1.122 & -0.271 & 0.107 & 1.121 \\
 & EB & -0.022 & 0.012 & 1.118 & -0.272 & 0.153 & 1.119 & -0.022 & 0.008 & 1.122 & -0.271 & 0.107 & 1.121 \\
 & IPW & -0.022 & 0.012 & 1.121 & -0.273 & 0.153 & 1.120 & -0.022 & 0.008 & 1.123 & -0.271 & 0.107 & 1.121 \\
\bottomrule
\end{tabular}%
}
\end{table}
Table~\ref{tab:s1-equivalence-total} shows close agreement among Aug., SBW, EB, and IPW. Within each scenario and sample size, their MC means and ESDs were nearly indistinguishable for both endpoints. Adjustment for the oracle score produced the largest gains: RE ranged from 1.360 to 1.368 at $n=400$ and from 1.343 to 1.352 at $n=800$. Adjustment using all five baseline covariates remained beneficial, with RE ranging from 1.250 to 1.258 at $n=400$ and from 1.260 to 1.264 at $n=800$. Over-adjustment, omission of key variables, and use of the wrong functional form attenuated but did not eliminate the efficiency gain. Stratified simple randomization gave RE values of 1.338--1.359, whereas the tighter balance induced by permuted blocks reduced the incremental gain from analysis-stage adjustment to approximately 1.12. 

For a further direct check of the bridge remainder in Theorem \ref{the:wLR_bridge}, let $\Delta_{h,n}=\widehat 
\psi_h-\widehat \psi_{\mathrm{aug}}$ and define the paired MC Root Mean Squared Error (RMSE) as $\sqrt{\tfrac{1}{10,000} \sum_{r=1}^{10,000} \Delta_{h,n,r}^2}$. Figure~\ref{fig:s1-paired-rmse-rate} plots this RMSE after multiplication by
$\sqrt n$.

\begin{figure}[htbp]
  \centering  \includegraphics[width=\linewidth]{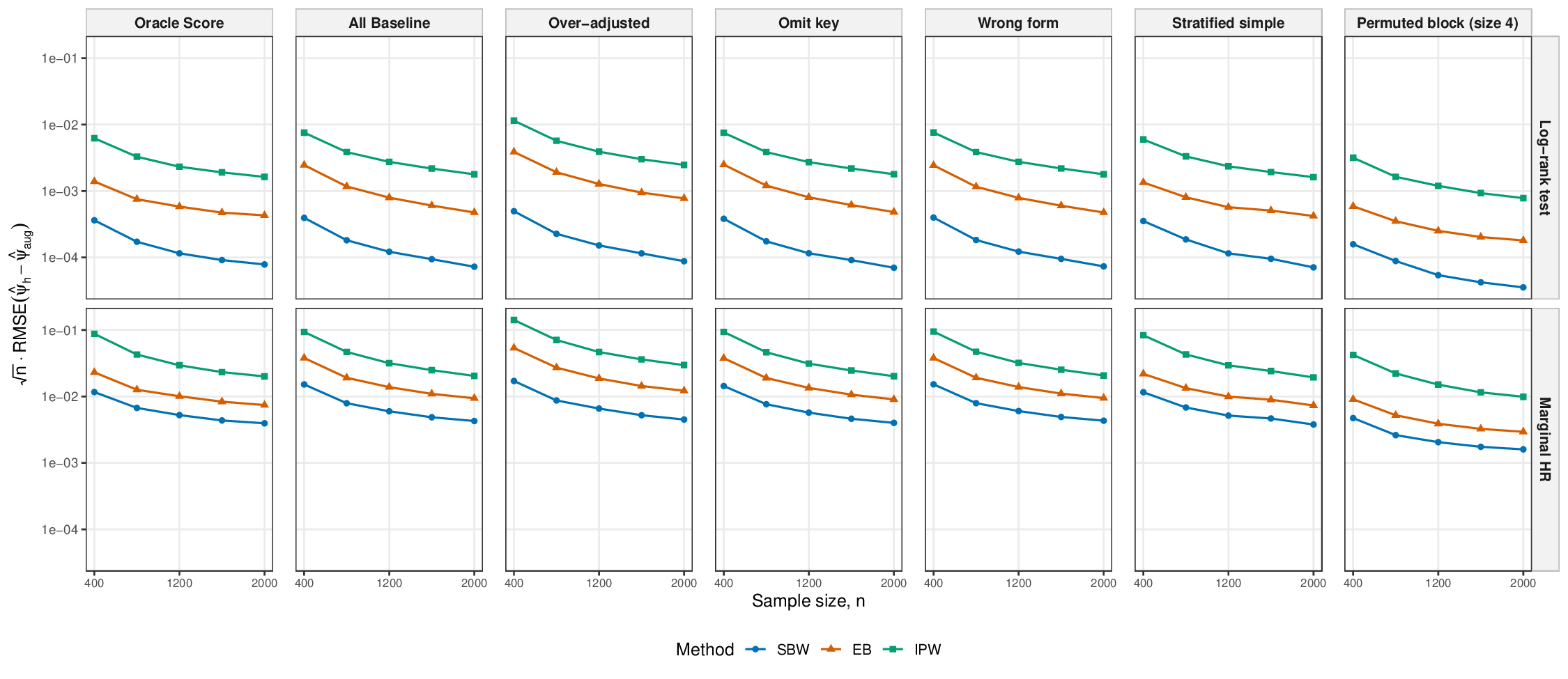}
  \caption{Paired RMSE-rate diagnostic for the bridge.  Curves show
$\sqrt n\,\{\mathbb{E}_{\mathrm{MC}}[(\widehat \psi_h-\widehat
\psi_{\mathrm{aug}})^2]\}^{1/2}$, for SBW, EB, and IPW across $n\in\{400,800,1200,1600,2000\}$, separately for the log-rank score and marginal log-HR.
All curves decrease monotonically; the scaled RMSE at $n=2000$ is 17.7\%--33.9\% of that at $n=400$. The vertical axis is logarithmic.}\label{fig:s1-paired-rmse-rate}
\end{figure}

Across all 42 combinations of scenario, endpoint, and weighting method, a
log-log regression over the five sample sizes for
$\mathrm{RMSE}_n\asymp n^{-\alpha}$ gave $\widehat\alpha=1.168$--$1.563$. The smallest lower 95\% confidence-limit was 1.017, well above the value $\alpha=1/2$ that a remainder of exact order $n^{-1/2}$ would produce, which supports the stated $o_p(n^{-1/2})$ bridge remainder.

\subsection{Simulation Study: Variance Estimation and Inference for the Marginal HR}\label{sec:Sim_Infer}

The inference study set $\beta_A=0$ and $\beta_Z=9$, used $n=400$, and included 10,000 MC replicates. Within each replicate, we generated one data set and held it fixed while applying the \emph{Oracle Score}, \emph{Over-adjusted}, \emph{Omit Key}, and \emph{Wrong Form} as working covariate specifications. Thus, these specifications changed only the covariates used in the fitted adjustment model; the generated covariates, treatment assignments, event times, and censoring outcomes were identical across specifications within a replicate. The unadjusted estimator does not depend on a working covariate specification, so it was computed once per replicate and reused as a common benchmark. Guided by the equivalence results and by the availability of method-specific variance estimators, we compared the unadjusted, augmented, EB, and IPW estimators. SBW was omitted because it lacks a packaged stacked M-estimation variance, which would preclude the same Primary-versus-bootstrap comparison \citep{shu2021variance,WeightIt}.

The ``Primary'' standard error was method-specific, chosen by what is available for each method.  For the unadjusted Cox model, it was the usual partial-likelihood information-based standard error returned by \texttt{coxph}; for the augmented estimator, the analytic standard error from \texttt{RobinCar} \citep{ye2024covariate}; and for EB and IPW, the stacked M-estimation standard error from \texttt{coxph\_weightit} in the R package \texttt{WeightIt} \citep{WeightIt}, which accounts for the uncertainty in the estimated weights in Section \ref{sec:VarEst}. The alternative standard error was obtained using a nonparametric bootstrap with $B=500$ subject-level resamples \citep{efron1994introduction}, refitting each method in every resample, including re-estimation of the EB and IPW weights. The ESD across the 10,000 MC point estimates served as an external benchmark for the mean estimated standard error.

\begin{table}[htbp]
\centering
\caption{Type I error and variance estimation for marginal HR inference with $n=400$. RR is the two-sided rejection rate ($\%$) at nominal
$\alpha=5\%$. The same 10,000 Monte Carlo data sets were analyzed under each
working covariate specification; the unadjusted estimator was computed once
per data set and is reported once as the common benchmark. A superscript
\(\dagger\) indicates an RR outside the 95\% MC reference interval $[4.57\%,\, 5.43\%]$.
}
\label{tab:type1-hr-primary-variance-ipw}
\resizebox{\textwidth}{!}{%
\begin{tabular}{llllccccc}
\toprule
\makecell[l]{Adjustment Set} & Method
& \makecell{RR (\%)\\\emph{Primary}} & \makecell{RR (\%)\\\emph{Bootstrap}}
& ESD
& \makecell{Mean SE\\\emph{Primary}} & \makecell{Mean SE\\\emph{Bootstrap}}
& \makecell{SE/ESD\\\emph{Primary}} & \makecell{SE/ESD\\\emph{Bootstrap}}\\
\midrule
None & Unadj. & 4.90 & 4.72 & 0.170 & 0.169 & 0.170 & 0.993 & 1.001 
\\\hdashline 
\addlinespace[2pt]
Oracle Score & Aug.  & 5.25 & 5.02 & 0.146 & 0.144 & 0.145 & 0.982 & 0.992
\\
 & EB  & 5.20 & 5.05 & 0.146 & 0.144 & 0.145 & 0.982 & 0.990
\\
 & IPW & 5.15 & 5.15 & 0.146 & 0.144 & 0.145 & 0.984 & 0.990
\\
\addlinespace[2pt]
Over-adjusted & Aug.  & 5.76\textsuperscript{\(\dagger\)} & 5.03 & 0.152 & 0.147 & 0.152 & 0.966 & 0.996
\\
 & EB  & 5.53\textsuperscript{\(\dagger\)} & 5.06 & 0.152 & 0.148 & 0.151 & 0.973 & 0.994
\\
 & IPW & 5.46\textsuperscript{\(\dagger\)} & 5.01 & 0.152 & 0.148 & 0.151 & 0.974 & 0.995
\\
\addlinespace[2pt]
Omit Key & Aug.  & 5.53\textsuperscript{\(\dagger\)} & 5.11 & 0.153 & 0.149 & 0.152 & 0.977 & 0.998
\\
 & EB  & 5.48\textsuperscript{\(\dagger\)} & 5.13 & 0.152 & 0.150 & 0.152 & 0.981 & 0.996
\\
 & IPW & 5.46\textsuperscript{\(\dagger\)} & 4.99 & 0.153 & 0.150 & 0.152 & 0.981 & 0.997
\\
\addlinespace[2pt]
Wrong Form & Aug.  & 5.51\textsuperscript{\(\dagger\)} & 5.04 & 0.151 & 0.147 & 0.150 & 0.974 & 0.995
\\
 & EB  & 5.44\textsuperscript{\(\dagger\)} & 5.24 & 0.151 & 0.147 & 0.150 & 0.978 & 0.993
\\
 & IPW & 5.45\textsuperscript{\(\dagger\)} & 5.04 & 0.151 & 0.147 & 0.150 & 0.978 & 0.993
\\
\bottomrule
\end{tabular}%
}
\end{table}

Table~\ref{tab:type1-hr-primary-variance-ipw} reports the common unadjusted
benchmark once, followed by the adjusted estimators under each working
covariate specification. With 10,000 replicates, an approximate 95\% MC
interval for a true 5\% rejection probability is $[4.57\%,\, 5.43\%]$. The
common unadjusted benchmark had Primary and bootstrap rejection rates of
4.90\% and 4.72\%, respectively, and corresponding SE/ESD ratios of 0.993
and 1.001. Under correct adjustment using the oracle score, the Primary
rejection rates were 5.25\%, 5.20\%, and 5.15\% for augmentation, EB, and IPW,
respectively; all were within the MC interval. Their mean Primary SEs were
98.2\%--98.4\% of the corresponding ESDs.

Under the \emph{Over-adjusted}, \emph{Omit Key}, and \emph{Wrong Form}
working covariate specifications, the Primary procedures for augmentation, EB, and IPW were mildly
anti-conservative. Their rejection rates ranged from 5.44\% to 5.76\%, with
mean SE/ESD ratios ranging from 0.966 to 0.981. This indicates modest
finite-sample underestimation of the sampling variability. The behavior of
IPW closely tracked that of EB: its Primary rejection rates were 5.46\%,
5.46\%, and 5.45\% under the three specifications, respectively.

In contrast, the bootstrap rejection rates for augmentation, EB, and IPW ranged from
4.99\% to 5.24\% across the more demanding working specifications, and all remained within the MC interval. The corresponding bootstrap SE/ESD ratios ranged from 0.993 to 0.998, across all four working specifications. Thus, the Primary analytic and M-estimation procedures were therefore accurate under correct adjustment
but exhibited modest finite-sample variance underestimation under more demanding adjustment sets. The explicit bootstrap provided the most uniform Type I error and variance calibration, while IPW and EB displayed nearly identical empirical behavior.

\subsection{Simulation Study: Prognostic strength and efficiency}

Prognostic strength and efficiency study fixed $n=400$ and $\beta_A=\log(0.7)$ and varied the prognostic
coefficient over $\beta_Z\in\{1,2,\ldots,10,12,15\}$ with total 12 levels. For each level, 8,000 MC  replicates were generated. As in the empirical equivalence study, this study also reports the variability through the ground-truth ESD. Efficiency was also compared for adjustment using either the oracle score $Z$ or the all baseline raw covariates in $X$.

The prognostic axes were calculated from a separate reference trial of 20,000
subjects. Within its control arm, Kendall's $\tau$ was the rank correlation
between $Z$ and the uncensored event time, and the C-index was the concordance
reported by a Cox model of the observed survival outcome on $Z$. Because larger
$Z$ implies shorter survival, the raw Kendall correlation is negative; the
figure displays $|\tau|$ as prognostic strength. Across the grid, $|\tau|$ increased from 0.059 to 0.528 and the C-index increased from 0.532 to 0.806.

\begin{figure}[htbp]
    \centering    \includegraphics[width=\linewidth]{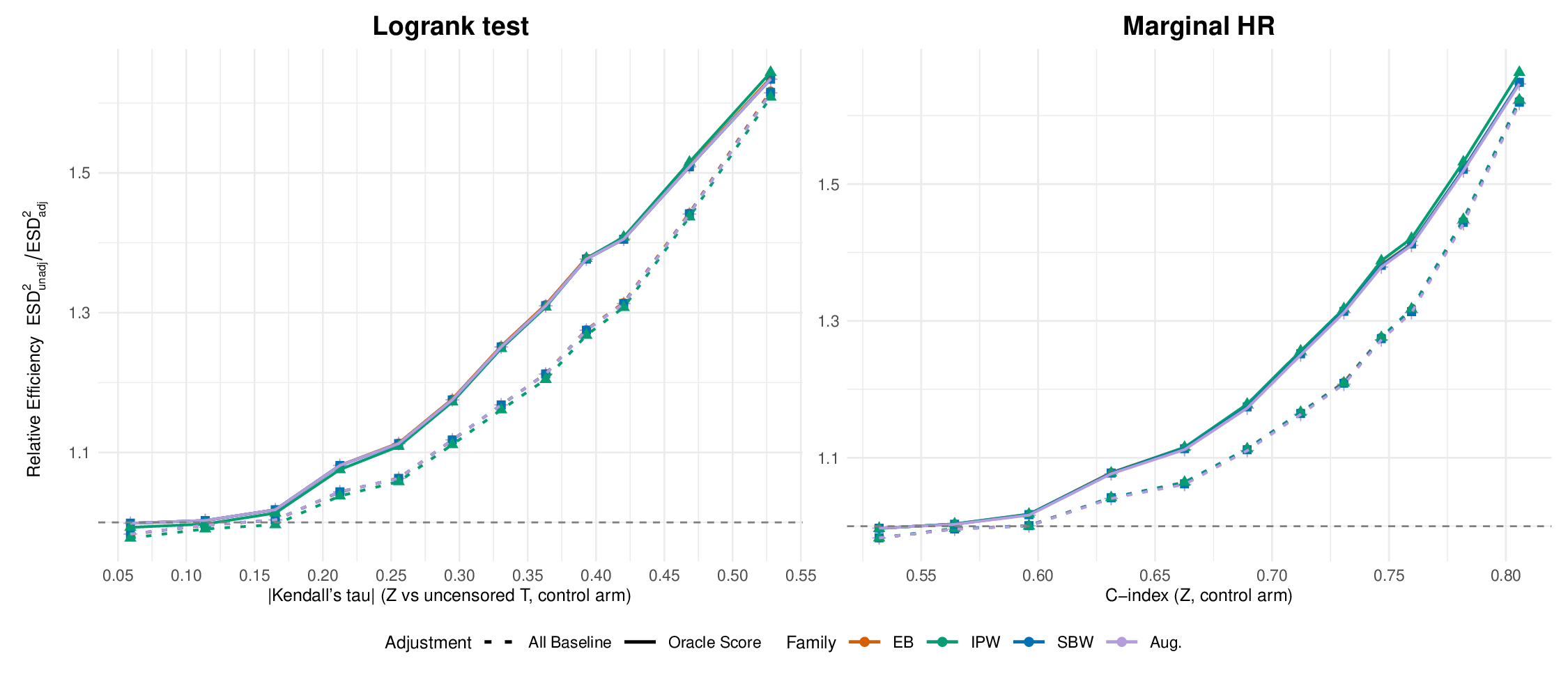}
    \caption{
    Relative efficiency as a function of prognostic strength. ESD is the empirical standard deviation across 8,000 MC replicates.
    }
    \label{fig:s2-prognostic-efficiency}
\end{figure}

Figure~\ref{fig:s2-prognostic-efficiency} shows that the efficiency gain increases monotonically with prognostic strength for both the log-rank test and the marginal Cox HR. Across both log rank test and marginal HR, augmentation and the balancing-weight implementations (EB, SBW, and IPW) yield nearly overlapping RE curves, supporting their empirical equivalence in this setting. When the prognostic signal is weak, the relative efficiency is close to one, indicating little gain from adjustment; as the C-index or Kendall's tau increases, the efficiency gain becomes substantial.

\section{Real Data Example}\label{sec:realdata}

We illustrate the proposed methods using the REWIND trial \citep{gerstein2019dulaglutide}, a multinational, randomized, double-blind, placebo-controlled cardiovascular outcomes trial of dulaglutide (1.5\,mg subcutaneous weekly) versus placebo in 9,901 participants with type~2 diabetes. Participants were aged 50 years or older with either established cardiovascular disease or cardiovascular risk factors, and the median follow-up was 5.4~years. Randomization was stratified by geographic region. The primary endpoint was the first occurrence of the three-component major adverse cardiovascular event (MACE-3): cardiovascular death, non-fatal myocardial infarction, or non-fatal stroke.

We applied three covariate-adjusted methods, including augmentation, EB, and IPW, to the MACE-3 endpoint, stratifying by geographic region. As in Section~\ref{sec:Sim_Infer}, SBW is omitted because a
packaged variance estimator is not available for it. Two adjustment sets were considered: (i) the AHA PREVENT 10-year cardiovascular risk score \citep{khan2024development}, a scalar composite of the baseline risk factors;  and
(ii) the individual PREVENT risk factors entered separately (age, sex, systolic blood pressure, antihypertensive use, total and HDL cholesterol, statin use, smoking, eGFR, and BMI).

\begin{table}[htbp]
\centering
\caption{Estimated marginal HRs for MACE-3 in the REWIND trial ($n=9{,}901$, 1{,}257 events). All adjusted analyses account for stratification
by geographic region. Variance reduction is
$1-\widehat{\Var}(\hat\theta)_{\mathrm{adj}}/\widehat{\Var}(\hat\theta)_{\mathrm{unadj}}$,
expressed as a percentage.}
\label{tab:rewind}
\begin{tabular}{llccc}
\toprule
Adjustment set & Method &$\exp{(\hat\theta)}$  (95\% CI) & $\mathrm{SE}(\hat\theta)$ & Var.\ reduction \\
\midrule
None & Unadj. & 0.882 (0.789, 0.985) & 0.0565 & --- \\
\addlinespace
PREVENT & Aug. & 0.884 (0.793, 0.986) & 0.0557 & 2.8\% \\
              & EB    & 0.884 (0.793, 0.986) & 0.0557 & 2.8\% \\
              & IPW   & 0.884 (0.793, 0.986) & 0.0557 & 2.8\% \\
\addlinespace
All Baseline & Aug. & 0.891 (0.799, 0.994) & 0.0555 & 3.5\% \\
                        & EB    & 0.896 (0.804, 0.999) & 0.0553 & 4.2\% \\
                        & IPW   & 0.891 (0.800, 0.994) & 0.0555 & 3.5\% \\
\bottomrule
\end{tabular}
\end{table}

Table~\ref{tab:rewind} reports the estimated marginal hazard ratios $\exp{(\hat\theta)}$. In the unadjusted analysis, dulaglutide reduced MACE-3 with $\exp{(\hat\theta)}=0.882$ (95\% CI: 0.789, 0.985). As expected under Theorem~\ref{thm:cox-equiv}, the three covariate-adjusted methods
give nearly identical estimates on this single dataset: adjusting for the PREVENT score, all three yield $\exp{(\hat\theta)}=0.884$ with identical confidence intervals; adjusting for the individual risk factors, they range from
0.891 (augmentation, IPW) to 0.896 (EB). 

The variance reduction column in Table~\ref{tab:rewind} reports the percentage reduction in the estimated variance of $\hat{\theta}$ relative to the unadjusted analysis. Adjustment for the PREVENT score reduced the estimated variance by 2.8\%, and adjustment for individual risk factors achieve 3.5--4.2\%, corresponding to standard errors of
0.0557 and 0.0553--0.0555 against 0.0565 unadjusted. 

The real-data findings reinforce two practical messages. First, the choice among augmentation, EB, and IPW is immaterial for the point estimate, so it can rest on computational convenience or interpretability; the consequential choice, as Section~\ref{sec:VarEst} discusses, is the variance estimator. Second, even in a well-powered trial whose unadjusted analysis is already significant, adjustment for pre-specified prognostic covariates reduces the estimated variance at no cost
in assumptions and without changing the estimand.

\section{Discussion}\label{sec:discussion}
We have shown that reweighting the arms for covariate balance and augmenting the score by within-arm regression are, to first order, the same covariate adjustment for time-to-event endpoints. The bridge is objective-free: any balancing-regular
weighting reproduces augmented score for both the log-rank test and the marginal hazard ratio, inheriting its validity and guaranteed efficiency gain while fitting no outcome model. Unlike the exact balancing--regression identity available for a linear estimator \citep{brunssmith2025augmented}, here the equivalence is only first-order, because the scores are nonlinear in the weights and the derived outcomes themselves depend on them; and it newly covers the covariate-adjusted log-rank test, which \citet{shao2026inverse} do not address.
 
Because the four available estimators agree to first order, the choice among them can rest on convenience; the consequential decisions are \emph{which} covariates to adjust for and how to perform inference. Although the weights are constructed without reference to the outcome, the efficiency gain comes entirely from the covariate--outcome association, so the selection rule is the same as for augmentation: balance on a few strongly prognostic covariates, pre-specified from external evidence \citep{fda2023covariate, vanlancker2024covariate}. The caveat is finite-sample, not asymptotic. Asymptotically the adjustment never loses efficiency, so a weakly associated covariate is neutral rather than harmful and merely shrinks the gain toward one. 

The Type~I error inflation we observed in Section~\ref{sec:Sim_Infer} is instead a variance-estimation issue, and it appeared nearly identically across the weighting
and augmentation estimators. It is consistent with the explanation offered in the literature, that with many covariates relative to the number of \emph{events} the analytic and M-estimation standard errors underestimate the true variability and the test over-rejects \citep{vanlancker2024covariate, tsiatis2008covariate}.
Because the point estimates stay accurate, the fix is on the variance side, a
bootstrap that re-estimates the weights, or a small-sample correction for the number of fitted parameters, restores near-nominal Type I error
\citep{austin2022bootstrap, ZhangVarEst2026, tsiatis2008covariate}. 
  
Several limitations bound these results. The equivalence is first-order and rests on first-moment balance, so it reproduces linear augmentation; it does not extend to the flexible or machine-learning augmentation of \citet{zhang2025unified}, for
which mean-balancing weights are not the design-side counterpart. We treat the Cox model as a working model, so when proportional hazards fails $\theta_0$ remains well defined but depends on the censoring distribution and the length of
follow-up, and its clinical interpretation is study-specific. In addition, we consider two arms and extension to multiple arms is conceptually straightforward but not developed here. Finally, SBW currently lacks a packaged variance estimator, which is why the inference study and the real-data analysis omit it.

Two directions extend the equivalence beyond first-moment balance. First, one can balance richer covariate functions than means, such as higher-order moments, interactions \citep{hainmueller2012entropy}, or a growing basis, in the limit a reproducing-kernel
Hilbert space \citep{wong2018kernel}, reproducing the richer augmentation that
approaches the nonparametric bound \citep{zhang2025unified, brunssmith2025augmented, athey2018approximate}.
As in the linear case, balancing and augmentation remain equivalent at every level
of richness \citep{brunssmith2025augmented}: a richer basis raises the attainable
efficiency gain but does not let weighting outperform the corresponding regression.

Second, balancing weights give a design-side, outcome-model-free way to adjust for a
pre-specified prognostic score, as in super-covariate or PROCOVA methods
\citep{holzhauer2023super, schuler2022increasing, ema2022procova}: the score is
predicted once from an external model and the trial arms are then balanced on it, which by our equivalence reproduces linear augmentation on the score while preserving the marginal survival estimand \citep{przybylski2026effect}. How much such a
score recovers relative to entering the risk factors separately will depend on how well the external model transfers to the trial population, and our REWIND analysis is a single dataset and cannot settle the question. When the external and trial populations differ, EB offers one route to reweighting the external data toward the trial before the score is fit \citep{josey2021transporting, chen2023entropy},  and quantifying the resulting gain is a natural next step.
 

\newpage
\appendix
\setcounter{equation}{0}
\renewcommand{\theequation}{S\arabic{equation}}
\setcounter{theorem}{0}
\renewcommand{\thetheorem}{S\arabic{theorem}}
\setcounter{section}{0}
\renewcommand{\thesection}{S\arabic{section}}

\part*{Supplementary Material}

\section{Equivalence of weighted log-rank test}

We do not attempt to analyze the weighted-score function $\Uh_h$ directly. Instead, we follow a \emph{proxy strategy}: identify a known, well-studied estimator that we expect to be
asymptotically equivalent to $\Uh_h$, prove the equivalence, and inherit the proxy's properties.
Here, the proxy is the augmented log-rank $\Uh_\mathrm{aug}$ of \cite{ye2024covariate}, for which the influence function, asymptotic variance, and
limiting distribution are well established. The bridge we will prove is
\begin{equation}\label{eq:strategy_bridge}
\sqrt n \bigl(\Uh_h - \Uh_\mathrm{aug}\bigr) \xrightarrow{p} 0.
\end{equation}
Once \eqref{eq:strategy_bridge} holds, Slutsky's theorem forces \(\sqrt n\,\Uh_h\)
and \(\sqrt n\,\Uh_\mathrm{aug}\) to share the same limiting distribution, so the weighted
log-rank test inherits both the validity and the guaranteed efficiency gain of
augmented test.

\paragraph{Roadmap.}
\begin{description}
  \item[Step 1 (Lemma~\ref{lem:identity}).]
  Establish an exact finite-sample representation of the weighted score
  \(\Uh_h\) in terms of weighted derived outcomes.
  \item[Step 2 (Lemmas~\ref{lem:ratio} and~\ref{lem:score_bridge}).]
  Linearize the perturbed risk-set ratio and obtain the first-order score
  expansion
  $
  \Uh_h-\Uh_\mathrm{unadj}
  =
  \frac1n\sum_{i=1}^n
  \xi_i\{A_i\Oh_{i1}-(1-A_i)\Oh_{i0}\}
  +\opn{n^{-1/2}}.$
  \item[Step 3 (Lemma~\ref{lem:car-mean-balance} and
  Propositions~\ref{prop:sbw-car}, \ref{prop:eb-car}, and
  \ref{prop:ipw-car}).]
  Verify that SBW, EB, and IPW are balancing-regular under the
  randomization condition.

  \item[Step 4 (Theorem~\ref{the:weighted-ye-bridge-car}).]
  Show that the leading score perturbation equals Ye's arm-specific
  regression adjustment up to \(\opn{n^{-1/2}}\), and hence
  $
  \Uh_h
  =
  \Uh_\mathrm{aug}
  +
  \opn{n^{-1/2}}.$
\end{description}

\subsection{Step 1 -- Weighted Algebraic Identity}
\begin{lemma}\label{lem:identity}
For any nonnegative weights, we have
\begin{equation}\label{eq:weightedidentity}
\Uh_h
\equiv
\frac{1}{n}\sum_{i=1}^n
h_i\{A_i\Oh_{i1}^{(h)}-(1-A_i)\Oh_{i0}^{(h)}\}.
\end{equation}
\end{lemma}
 
\begin{proof}
\medskip
\noindent\textbf{Stage A: arm-symmetric form of \(\Uh_h\).}
 
Split the integrand in $\hat{U}_h$ into two pieces by distributivity:
\[
\Uh_h
= \frac{1}{n}\sum_{i=1}^n h_iA_i\int_0^\tau dN_i(t)
-\frac{1}{n}\sum_{i=1}^n h_i\int_0^\tau
\frac{\Yb_1^{(h)}(t)}{\Yb^{(h)}(t)}dN_i(t)=:\mathcal{A}-\mathcal{B}.
\]
 
\smallskip
\emph{Term} $\mathcal{A}$.
Since $
A_i\,dN_i(t)
=
A_i\{A_i\,dN_{i1}(t)+(1-A_i)\,dN_{i0}(t)\}
=
A_i\,dN_{i1}(t).
$ Then, 
\[
\mathcal{A}
=
\frac{1}{n}\sum_{i=1}^n h_iA_i\int_0^\tau dN_{i1}(t)
=
\int_0^\tau\frac{1}{n}\sum_{i:A_i=1}h_i\,dN_{i1}(t)
=
\int_0^\tau d\Nb_1^{(h)}(t).
\]
 
\smallskip
\emph{Term} $\mathcal{B}$.
The ratio \(\Yb_1^{(h)}(t)/\Yb^{(h)}(t)\) is independent of the index
\(i\) and can be moved outside the sum:
\[
\mathcal{B}
=
\int_0^\tau
\frac{\Yb_1^{(h)}(t)}{\Yb^{(h)}(t)}
\;\frac{1}{n}\sum_{i=1}^n h_i\,dN_i(t)
=
\int_0^\tau
\frac{\Yb_1^{(h)}(t)}{\Yb^{(h)}(t)}\,d\Nb^{(h)}(t)
\] 
\smallskip
Decompose
\(d\Nb^{(h)}=d\Nb_1^{(h)}+d\Nb_0^{(h)}\).
Combine the first two integrals, and use
\(1-\Yb_1^{(h)}/\Yb^{(h)}=\Yb_0^{(h)}/\Yb^{(h)}\):
\begin{equation}\label{eq:Uharm}
\begin{aligned}
    \Uh_h
    &=
    \int_0^\tau d\Nb_1^{(h)}(t)
    -
    \int_0^\tau\frac{\Yb_1^{(h)}(t)}{\Yb^{(h)}(t)}d\Nb_1^{(h)}(t)
    -
    \int_0^\tau\frac{\Yb_1^{(h)}(t)}{\Yb^{(h)}(t)}d\Nb_0^{(h)}(t) \\
    &=
    \int_0^\tau
    \frac{\Yb_0^{(h)}(t)}{\Yb^{(h)}(t)}
    \,d\Nb_1^{(h)}(t)
    -
    \int_0^\tau
    \frac{\Yb_1^{(h)}(t)}{\Yb^{(h)}(t)}
    \,d\Nb_0^{(h)}(t).
\end{aligned}
\end{equation}
\noindent\textbf{Stage B: arm-symmetric form of the right-hand side (Derived Outcome diff).}
First consider the contribution from the treated arm. By the
definition of \(\Oh_{i1}^{(h)}\),
\begin{equation}
    \begin{aligned}
      \frac{1}{n}\sum_{i=1}^n h_iA_i\Oh_{i1}^{(h)}
        &=
        \frac{1}{n}\sum_{i:A_i=1} h_i
        \int_0^\tau
        \frac{\Yb_0^{(h)}(t)}{\Yb^{(h)}(t)}
        \left\{
        dN_{i1}(t)
        -
        Y_{i1}(t)\frac{d\Nb^{(h)}(t)}{\Yb^{(h)}(t)}
        \right\}\\
        &=
        \int_0^\tau
        \frac{\Yb_0^{(h)}(t)}{\Yb^{(h)}(t)}
        d\Nb_1^{(h)}(t)
        -
        \int_0^\tau
        \frac{\Yb_0^{(h)}(t)}{\Yb^{(h)}(t)}
        \frac{\Yb_1^{(h)}(t)}{\Yb^{(h)}(t)}
        d\Nb^{(h)}(t),
    \end{aligned}
\end{equation}
Similarly, for the control arm,
\[
\frac{1}{n}\sum_{i=1}^n h_i(1-A_i)\Oh_{i0}^{(h)}
=
\int_0^\tau
\frac{\Yb_1^{(h)}(t)}{\Yb^{(h)}(t)}
d\Nb_0^{(h)}(t)
-
\int_0^\tau
\frac{\Yb_1^{(h)}(t)}{\Yb^{(h)}(t)}
\frac{\Yb_0^{(h)}(t)}{\Yb^{(h)}(t)}
d\Nb^{(h)}(t).
\]
 
Taking the treated contribution minus the control contribution yields the right part of \eqref{eq:Uharm}.  
\end{proof}

\subsection{Step 2 -- Risk-set ratio Linearization and Score Difference Representation}
Let weighted and unweighted risk-sets ratio denote as
$
R^{(h)}(t)=\frac{\Yb_1^{(h)}(t)}{\Yb^{(h)}(t)},$ and $
R(t)=\frac{\Yb_1(t)}{\Yb(t)}.$

\begin{lemma}[Ratio linearization]\label{lem:ratio}
Assume, for \(j=0,1\),
\begin{equation}\label{eq:xi^2}
   \frac1n\sum_{i:A_i=j}\xi_i^2=O_p(n^{-1}).
\end{equation} 
Then, uniformly over $t\in[0,\tau]$,
\[
R^{(h)}(t)-R(t) =
\frac{1}{\Yb(t)}
\frac1n\sum_{i=1}^n
\xi_i\{A_i-R(t)\}Y_i(t)
+
r_n(t),
\qquad
\sup_{t\le\tau}|r_n(t)|=\opn{n^{-1/2}}.
\]
In particular,
$
\sup_{t\le\tau}|R^{(h)}(t)-R(t)|=\Opn{n^{-1/2}}.
$
\end{lemma}

\begin{proof}
Define the two perturbations
\[
D_1(t) = \Yb_1^{(h)}(t) - \Yb_1(t) = \frac{1}{n}\sum_{i=1}^n A_i \xi_i Y_i(t), \quad
D(t) = \Yb^{(h)}(t) - \Yb(t) = \frac{1}{n}\sum_{i=1}^n \xi_i Y_i(t).
\]

\noindent\textbf{Part A. Perturbations 
bounds.}
By Cauchy--Schwarz and \(0\le Y_i(t)\le1\),
\[
\sup_{t\le\tau}|D_j(t)|
=
\sup_{t\le\tau}
\Big|
\frac1n\sum_{i:A_i=j}\xi_iY_i(t)
\Big|
\le
\Big\{
\frac1n\sum_{i:A_i=j}\xi_i^2
\Big\}^{1/2}
=O_p(n^{-1/2}),
\]
for \(j=0,1\). Therefore,
$
\sup_{t\le\tau}|D(t)|
\le
\sup_{t\le\tau}|D_0(t)|
+
\sup_{t\le\tau}|D_1(t)|
=O_p(n^{-1/2}).
$

\noindent\textbf{Part B. Exact algebraic decomposition.}
\noindent A direct algebraic calculation gives the exact identity
\begin{align*}
R^{(h)}(t)-R(t)
&=\frac{\{\Yb_1(t)+D_1(t)\}\Yb(t)-\Yb_1(t)\{\Yb(t)+D(t)\}}
        {\Yb^{(h)}(t)\Yb(t)} =\frac{D_1(t)-R(t)D(t)}{\Yb^{(h)}(t)}.
\end{align*}
Next, by def of $D(t)$ we have
$
\frac1{\Yb^{(h)}(t)}=\frac1{\Yb(t)}-\frac{D(t)}{\Yb(t)\Yb^{(h)}(t)}
$, so
\begin{equation*}\label{eq:ratio_decomp_main}
R^{(h)}(t)-R(t)
= \frac{D_1(t)-R(t)D(t)}{\Yb(t)}
- \frac{\{D_1(t)-R(t)D(t)\}D(t)}{\Yb(t)\Yb^{(h)}(t)}: =\mathcal{A}(t)- \mathcal{B}(t).
\end{equation*}
By definition,
$
\mathcal{A}(t)=\frac1{\Yb(t)}\frac1n\sum_{i=1}^n\xi_i\{A_i-R(t)\}Y_i(t).
$

\noindent\textbf{Part C. Uniform bound for the remainder $\mathcal{B}(t)$.}
Since $r_n(t)=-\mathcal{B}(t)$, it suffices to bound $\mathcal{B}(t)$; we treat
the reciprocal denominator and the numerator separately.

\medskip
\noindent\emph{Step 1: the reciprocal denominator is $\Opn{1}$.}
Let $m(t)=\E\{Y_i(t)\}$. By the Glivenko--Cantelli theorem for the monotone
class $\{Y_i(t):t\le\tau\}$, $\sup_{t\le\tau}|\Yb(t)-m(t)|=\opn{1}$. Since
$\inf_{t\le\tau}m(t)>c>0$, for every $t\le\tau$
\[
\Yb(t)\ge m(t)-|\Yb(t)-m(t)|
\ge \inf_{s\le\tau}m(s)-\sup_{s\le\tau}|\Yb(s)-m(s)|,
\]
so that $\inf_{t\le\tau}\Yb(t)\ge c-\sup_{t\le\tau}|\Yb(t)-m(t)|$. Hence
$\{\inf_{t\le\tau}\Yb<c/2\}\subseteq\{\sup_{t\le\tau}|\Yb-m|>c/2\}$, and since the
right-hand event has probability tending to $0$,
\begin{equation}\label{eq:Ybar_lb}
\Pr\Big\{\inf_{t\le\tau}\Yb(t)\ge c/2\Big\}\to1.
\end{equation}
Next, $\Yb^{(h)}(t)=\Yb(t)+D(t)\ge\Yb(t)-|D(t)|$, so
$\inf_{t\le\tau}\Yb^{(h)}(t)\ge\inf_{t\le\tau}\Yb(t)-\sup_{t\le\tau}|D(t)|$ with
$\sup_{t\le\tau}|D(t)|=\Opn{n^{-1/2}}=\opn{1}$. If both
$\inf_{t\le\tau}\Yb\ge c/2$ and $\sup_{t\le\tau}|D|\le c/4$ hold, then
$\Yb^{(h)}(t)\ge c/2-c/4=c/4$ for all $t\le\tau$; equivalently
$\Big\{\inf_{t\le\tau}\Yb^{(h)}<c/4\Big\}
\subseteq
\Big\{\inf_{t\le\tau}\Yb<c/2\Big\}\cup\Big\{\sup_{t\le\tau}|D|>c/4\Big\}.
$
Both events on the right have probability tending to $0$ (by \eqref{eq:Ybar_lb}
and $\sup_{t\le\tau}|D|=\opn{1}$), so the union bound gives
\begin{equation}\label{eq:Ybarh_lb}
\Pr\Big\{\inf_{t\le\tau}\Yb^{(h)}(t)\ge c/4\Big\}\to1.
\end{equation}

Write $A_n=\{\inf_{t\le\tau}\Yb\ge c/2\}$ and
$B_n=\{\inf_{t\le\tau}\Yb^{(h)}\ge c/4\}$. On $A_n\cap B_n$ we have
$\inf_{t\le\tau}\Yb(t)\Yb^{(h)}(t)\ge(c/2)(c/4)=c^2/8$, and by
\eqref{eq:Ybar_lb}--\eqref{eq:Ybarh_lb},
$\Pr(A_n\cap B_n)\ge1-\Pr(A_n^c)-\Pr(B_n^c)\to1$. Therefore
\begin{equation}\label{eq:denomOp1}
\begin{aligned}
&\Pr(A_n\cap B_n)\to1\\
&\quad\implies\ \Pr\Big\{\inf_{t\le\tau}\Yb(t)\Yb^{(h)}(t)\ge c^2/8\Big\}\to1
&&\text{[bound holds on $A_n\cap B_n$]}\\
&\quad\iff\ \Pr\Big\{\sup_{t\le\tau}\tfrac{1}{\Yb(t)\Yb^{(h)}(t)}\le 8/c^2\Big\}\to1
&&\text{[reciprocal of a positive quantity]}\\
&\quad\implies\ \sup_{t\le\tau}\frac{1}{\Yb(t)\Yb^{(h)}(t)}=\Opn{1}
&&\text{[definition of $\Opn{1}$].}
\end{aligned}
\end{equation}

\smallskip
\noindent\emph{Step 2: the numerator is $\Opn{n^{-1}}$.}
Since $0\le R(t)\le1$, for every $t\le\tau$,
\[
|D_1(t)-R(t)D(t)|\le|D_1(t)|+R(t)|D(t)|\le|D_1(t)|+|D(t)|,
\]
so, taking the supremum and using
$\sup_{t\le\tau}|D_1(t)|=\Opn{n^{-1/2}}$,
$\sup_{t\le\tau}|D(t)|=\Opn{n^{-1/2}}$ from Part~A,
\[
\sup_{t\le\tau}|D_1(t)-R(t)D(t)|
\le\sup_{t\le\tau}|D_1(t)|+\sup_{t\le\tau}|D(t)|=\Opn{n^{-1/2}}.
\]
Multiplying by $\sup_{t\le\tau}|D(t)|=\Opn{n^{-1/2}}$,
\begin{equation}\label{eq:numbound}
\sup_{t\le\tau}\bigl|\{D_1(t)-R(t)D(t)\}D(t)\bigr|
\le\sup_{t\le\tau}|D_1(t)-R(t)D(t)|\;\sup_{t\le\tau}|D(t)|
=\Opn{n^{-1}}.
\end{equation}

\medskip
\noindent\emph{Step 3: combine numerator and denominator.}
Combining \eqref{eq:denomOp1} and \eqref{eq:numbound},
\begin{equation*}\label{eq:Bbound}
\begin{aligned}
\sup_{t\le\tau}|\mathcal{B}(t)|
&\le
\Big\{\sup_{t\le\tau}\frac{1}{\Yb(t)\Yb^{(h)}(t)}\Big\}
\Big\{\sup_{t\le\tau}\bigl|\{D_1(t)-R(t)D(t)\}D(t)\bigr|\Big\}
&&\text{[product of sup]}\\
&=\Opn{1}\cdot\Opn{n^{-1}}=\Opn{n^{-1}}=\opn{n^{-1/2}}
&&\text{[since $n^{-1}/n^{-1/2}\to0$].}
\end{aligned}
\end{equation*}
Since $r_n(t)=-\mathcal{B}(t)$, this proves
$\sup_{t\le\tau}|r_n(t)|=\opn{n^{-1/2}}$.
\end{proof}

\begin{lemma}[Derived-outcome representation of the score perturbation]
\label{lem:score_bridge}
Assume the conditions of Lemma~\ref{lem:ratio}.
Then
\begin{equation}\label{eq:scoreperturbation}
\Uh_h-\Uh_\mathrm{unadj}
=\frac1n\sum_{i=1}^n\xi_i\{A_i\Oh_{i1}-(1-A_i)\Oh_{i0}\}+\opn{n^{-1/2}}.
\end{equation}
\end{lemma}

\begin{proof}
From the definitions of the $\hat{U}_h$ and $\hat{U}_L$, we have
\begin{equation}\label{eq:score_diff_start}
\Uh_h-\Uh_\mathrm{unadj}
=\frac1n\sum_{i=1}^n\int_0^\tau
\bigl[h_i\{A_i-R^{(h)}(t)\}-\{A_i-R(t)\}\bigr]\,dN_i(t).
\end{equation}
Using \(h_i=1+\xi_i\), the integrand decomposes exactly as
\begin{equation}\label{eq:integrand_decomp}
h_i\{A_i-R^{(h)}\}-\{A_i-R\}
=(h_i-1)\{A_i-R\}+h_i\{R-R^{(h)}\}
=\xi_i\{A_i-R\}-h_i\{R^{(h)}-R\}.
\end{equation}
Substituting \eqref{eq:integrand_decomp} into \eqref{eq:score_diff_start},
\begin{equation}\label{eq:score_exact_decomp}
\Uh_h -\Uh_\mathrm{unadj}
= \frac1n\sum_{i=1}^n \xi_i \int_0^\tau \{A_i-R(t)\}\,dN_i(t)
-\frac1n\sum_{i=1}^n h_i\int_0^\tau\{R^{(h)}(t)-R(t)\}\,dN_i(t):=\mathcal{A}_n - \mathcal{B}_n.
\end{equation}
Split $\mathcal{B}_n^{(h)}$ into two parts:
\[
\mathcal{B}_n^{(h)}
=\frac1n\sum_i\int_0^\tau\{R^{(h)}-R\}\,dN_i
+\frac1n\sum_i\xi_i\int_0^\tau\{R^{(h)}-R\}\,dN_i.
\]
We first show the second part is negligible. By Lemma~\ref{lem:ratio},
$\sup_{t\le\tau}|R^{(h)}(t)-R(t)|=\Opn{n^{-1/2}}$. Since $N_i(\tau)\le1$ for a
single TTE, $\frac1n\sum_i N_i(\tau)^2\le\frac1n\sum_i N_i(\tau)\le1$. So by Cauchy--Schwarz and \eqref{eq:xi^2}
,
\[
\frac1n\sum_{i=1}^n|\xi_i|N_i(\tau)
\le\Big(\frac1n\sum_{i=1}^n\xi_i^2\Big)^{1/2}
   \Big(\frac1n\sum_{i=1}^n N_i(\tau)^2\Big)^{1/2}
\le\Big(\frac1n\sum_{i=1}^n\xi_i^2\Big)^{1/2}=\Opn{n^{-1/2}}.
\]
Therefore
\[
\Big|\frac1n\sum_i\xi_i\int_0^\tau\{R^{(h)}-R\}\,dN_i\Big|
\le\sup_{t\le\tau}|R^{(h)}-R|\;\frac1n\sum_i|\xi_i|N_i(\tau)
=\Opn{n^{-1/2}}\cdot\Opn{n^{-1/2}}=\Opn{n^{-1}}=\opn{n^{-1/2}},
\]
and hence
\begin{equation}\label{eq:B_unweighted_main}
\Uh_h-\Uh_\mathrm{unadj}
=\mathcal{A}_n-\int_0^\tau\{R^{(h)}(t)-R(t)\}\,d\Nb(t)+\opn{n^{-1/2}}.
\end{equation}
By Lemma~\ref{lem:ratio},
and since $\Nb(\tau)\le1$,
$\big|\int_0^\tau r_n\,d\Nb\big|\le\sup_{t\le\tau}|r_n(t)|\,\Nb(\tau)=\opn{n^{-1/2}}$, so
\begin{equation}\label{eq:B_expand_main}
\int_0^\tau\{R^{(h)}(t)-R(t)\}\,d\Nb(t)
=\frac1n\sum_{i=1}^n\xi_i\int_0^\tau\frac{\{A_i-R(t)\}Y_i(t)}{\Yb(t)}\,d\Nb(t)+\opn{n^{-1/2}}.
\end{equation}
Combining \eqref{eq:B_unweighted_main} and \eqref{eq:B_expand_main},
\[
\Uh_h-\Uh_\mathrm{unadj}
=\frac1n\sum_i\xi_i
\left[\int_0^\tau\{A_i-R(t)\}\,dN_i(t)
-\int_0^\tau\frac{\{A_i-R(t)\}Y_i(t)}{\Yb(t)}\,d\Nb(t)\right]+\opn{n^{-1/2}}.
\]
In the bracket, If $A_i=1$, then $N_i=N_{i1}$, $Y_i=Y_{i1}$,
and $A_i-R(t)=\Yb_0(t)/\Yb(t)$, so the bracket equals
\[
\int_0^\tau\frac{\Yb_0(t)}{\Yb(t)}
\Big\{dN_{i1}(t)-Y_{i1}(t)\frac{d\Nb(t)}{\Yb(t)}\Big\}=\Oh_{i1}.
\]
If $A_i=0$, then $N_i=N_{i0}$, $Y_i=Y_{i0}$, and $A_i-R(t)=-\Yb_1(t)/\Yb(t)$, so
the bracket equals
\[
-\int_0^\tau\frac{\Yb_1(t)}{\Yb(t)}
\Big\{dN_{i0}(t)-Y_{i0}(t)\frac{d\Nb(t)}{\Yb(t)}\Big\}=-\Oh_{i0}.
\]
This proves \eqref{eq:scoreperturbation}.
\end{proof}

\subsection{Step 3 -- Weight expansions and stability under stratified randomization}
\label{subsec:car-step3}

\begin{remark}[Designs satisfying the Randomization condition]
\label{rem:common-car}
The randomization condition covers simple randomization as the single-stratum
case \(|\mathcal Z|=1\), and stratified randomization on the levels of \(Z\). We verify the rate \(n_{z1}/n_z-\pi=\Opn{n^{-1/2}}\) for the stratified schemes.
Under simple (or stratified simple) randomization, the treated
count in level \(z\) is \(n_{z1}\sim\mathrm{Bin}(n_z,\pi)\) conditionally on
the stratum sizes, so the central limit theorem gives
\(n_{z1}/n_z-\pi=\Opn{n^{-1/2}}\).  Under stratified permuted-block
randomization with maximal block size \(C_B\), the imbalance in each level comes
only from one incomplete block, so \(|n_{z1}-\pi n_z|\le C_B\) deterministically
and, since \(n_z/n\to_p\Pr(Z=z)>0\),
\[
\Big|\frac{n_{z1}}{n_z}-\pi\Big|\le\frac{C_B}{n_z}=\Opn{n^{-1}},
\]
which is stronger than required.  
\end{remark}

\begin{lemma}[Stratified covariate-mean balance and within-arm covariance consistency]
\label{lem:car-mean-balance}
\label{lem:car-cov-consistency}
Under the randomization condition, for \(j=0,1\):
\begin{enumerate}[label=\textup{(\alph*)},leftmargin=2.2em]
\item \emph{(mean balance)}
\begin{equation}\label{eq:car_mean_balance}
\Xb_j-\Xb=\Opn{n^{-1/2}};
\end{equation}
\item \emph{(covariance consistency)}
\begin{equation}\label{eq:car_cov_consistency}
\Sh_{X|j}
=\frac1{n_j}\sum_{i:A_i=j}(X_i-\Xb_j)(X_i-\Xb_j)^\top
\xrightarrow{p}\Sigma_X,
\end{equation}
and the pooled, grand-mean-centered covariance
\(\hat\Sigma_X\) satisfies
\(\hat\Sigma_X\xrightarrow{p}\Sigma_X\), so \(\hat\Sigma_X-\Sh_{X|j}=\opn{1}\).
\end{enumerate}
\end{lemma}
 
\begin{proof}
We argue for \(j=1\); arm \(0\) follows by replacing \(\pi\) with \(1-\pi\).
For each \(z\in\mathcal Z\) write
\(n_z=\sum_i\mathbf 1(Z_i=z)\), \(n_{z1}=\sum_i A_i\mathbf 1(Z_i=z)\),
\(\Xb_z=n_z^{-1}\sum_{i:Z_i=z}X_i\), and
\(\Xb_{z1}=n_{z1}^{-1}\sum_{i:Z_i=z,A_i=1}X_i\).
 
\medskip
\noindent\textbf{Shared fact}
Within each stratum \(z\), treatment labels are assigned without using
covariate values beyond the stratification variables. Hence, for any
\(\phi\) with finite second moment,
\begin{equation}\label{eq:car_repr}
\frac1{n_{z1}}\sum_{i:Z_i=z,A_i=1}\phi(X_i)
-
\frac1{n_z}\sum_{i:Z_i=z}\phi(X_i)
=
\Opn{n_z^{-1/2}}.
\end{equation}
 
\medskip
\noindent\textbf{Part (a): mean balance.}
Both \(\Xb_1\) and \(\Xb\) are stratum-weighted averages of stratum quantities, $\Xb_1=\sum_{z\in\mathcal Z}\frac{n_{z1}}{n_1}\Xb_{z1}, \Xb=\sum_{z\in\mathcal Z}\frac{n_z}{n}\Xb_z.$ Writing \(\Xb_{z1}=\Xb_z+(\Xb_{z1}-\Xb_z)\) in the first sum and subtracting,
\[
\Xb_1-\Xb
=\sum_{z\in\mathcal Z}\Big(\frac{n_{z1}}{n_1}-\frac{n_z}{n}\Big)\Xb_z
+\sum_{z\in\mathcal Z}\frac{n_{z1}}{n_1}(\Xb_{z1}-\Xb_z),
\]
where the first sum collects the within-stratum mean \(\Xb_z\) against the weight difference, and the second collects the treated-vs-stratum deviation. For the first sum, \(\frac{n_{z1}}{n_1} - \frac{n_z}{n}
= \frac{n_z}{n}\frac{(n_{z1}/n_z) - (n_1/n)}{n_1/n}\); since both \(n_{z1}/n_z - \pi\) and \(n_1/n - \pi\) are \(\Opn{n^{-1/2}}\) while
\(n_1/n\to_p\pi>0\), each term is \(\Opn{n^{-1/2}}\), and with \(\Xb_z=\Opn{1}\)
(as \(\E\|X_i\|^2<\infty\)) and \(\mathcal Z\) finite the sum is
\(\Opn{n^{-1/2}}\).  For the second sum, \eqref{eq:car_repr} with \(\phi(x)=x\) gives
$ \Xb_{z1}-\Xb_z=\Opn{n_z^{-1/2}}=\Opn{n^{-1/2}}, $
because \(n_z/n\to_p p_z>0\). Since \(n_{z1}/n_1=\Opn1\) and
\(\mathcal Z\) is finite, the second sum is \(\Opn(n^{-1/2})\).
Combining the two sums proves \eqref{eq:car_mean_balance}.
 
\medskip
\medskip
\noindent\textbf{Part (b): covariance consistency.}
By \eqref{eq:car_repr} with \(\phi(x)=x\), $\Xb_{z1}-\Xb_z=\opn1$ for each \(z\).  Applying the same representative-subsample argument
componentwise to \(\phi(x)=xx^\top\), with truncation if needed under
\(\E\|X_i\|^2<\infty\), gives
\[
\frac1{n_{z1}}\sum_{i:Z_i=z,A_i=1}X_iX_i^\top
-
\frac1{n_z}\sum_{i:Z_i=z}X_iX_i^\top
=
\opn1 .
\]
Since \(n_{z1}/n_z\to_p\pi\) and \(\mathcal Z\) is finite, aggregating over
strata gives
$ \Xb_1\xrightarrow{p}\E X_i,$ and $
\frac1{n_1}\sum_{i:A_i=1}X_iX_i^\top
\xrightarrow{p}
\E(X_iX_i^\top).$
Therefore
\[
\Sh_{X|1}
=
\frac1{n_1}\sum_{i:A_i=1}X_iX_i^\top-\Xb_1\Xb_1^\top
\xrightarrow{p}
\E(X_iX_i^\top)-\E X_i\,\E X_i^\top
=
\Sigma_X.
\]
The pooled covariance is a full-sample average and converges to the same limit by the law of large numbers. Hence $\hat\Sigma_X-\Sh_{X|j}=\opn1.$
The argument for arm \(0\) is identical.
\end{proof}
\vspace{20pt}

Calibrated weights \(h_i=1+\xi_i\) satisfy,
within each arm \(j=0,1\),
\begin{equation}\label{eq:hconstraints-car}
\frac1{n_j}\sum_{i:A_i=j}h_i=1,
\qquad
\frac1{n_j}\sum_{i:A_i=j}h_iX_i=\Xb,
\qquad h_i\ge0,
\end{equation}
\begin{proposition}[SBW is Balancing-regular under randomization conditions]
\label{prop:sbw-car}
For each arm \(j=0,1\), define the SBW by
\begin{equation}\label{eq:sbw_objective_car}
\min_{\{h_i:A_i=j\}}\frac1{2n_j}\sum_{i:A_i=j}(h_i-1)^2
\quad\text{subject to }\eqref{eq:hconstraints-car}.
\end{equation}
Assume randomization condition.  Then, with probability tending to
one,
\begin{equation}\label{eq:sbw_xiexpansion_car}
\xi_i=(X_i-\Xb_j)^\top\Sh_{X|j}^{-1}(\Xb-\Xb_j),
\qquad i:A_i=j,
\end{equation}
i.e. balancing-regular weight holds with remainder
\(\rho_i\equiv0\), and
\begin{equation}\label{eq:sbw_xisquare_car}
\frac1n\sum_{i:A_i=j}\xi_i^2=\Opn{n^{-1}},
\qquad j=0,1.
\end{equation}
\end{proposition}
 
\begin{proof}
Assume that the non-negativity constraints are inactive with probability tending to one. 
 
\medskip\noindent\textbf{Step 1 (Lagrangian and first-order conditions).}
The Lagrangian for arm \(j\) is
\[
\mathcal{L}_j(\xi_i;\alpha_j,\eta_j)
=
\frac1{2n_j}\sum_{i:A_i=j}\xi_i^2
-\alpha_j\Big
(\frac1{n_j}\sum_{i:A_i=j}\xi_i\Big)
-\eta_j^\top\Big\{\frac1{n_j}\sum_{i:A_i=j}\xi_i(X_i-\Xb_j)-(\Xb-\Xb_j)\Big\}.
\]
The first-order condition with respect to \(\xi_i\) gives
$\xi_i=\alpha_j+\eta_j^\top(X_i-\Xb_j).
$ Then, averaging over \(i:A_i=j\) and using
\(n_j^{-1}\sum_{i:A_i=j}(X_i-\Xb_j)=0\) and
\(n_j^{-1}\sum_{i:A_i=j}\xi_i=0\) yields \(\alpha_j=0\).  Hence
\[
\xi_i=\eta_j^\top(X_i-\Xb_j).
\]
 
\medskip\noindent\textbf{Step 2 (Closed form for \(\eta_j\)).}
Substituting \(\xi_i=\eta_j^\top(X_i-\Xb_j)\) into the covariate-balance
constraint gives
\[
\frac1{n_j}\sum_{i:A_i=j}\{\eta_j^\top(X_i-\Xb_j)\}(X_i-\Xb_j)
=\frac1{n_j}\sum_{i:A_i=j}(X_i-\Xb_j)(X_i-\Xb_j)^\top\eta_j
=\Xb-\Xb_j,
\]
so by the definition of \(\Sh_{X|j}\),
\[
\Sh_{X|j}*\eta_j=\Xb-\Xb_j .
\]
By Lemma~\ref{lem:car-cov-consistency}, \(\Sh_{X|j}\overset{p}{\to}\Sigma_X\),
and since \(\lambda_{\min}\) is continuous in the matrix entries,
\(\lambda_{\min}(\Sh_{X|j})\overset{p}{\to}\lambda_{\min}(\Sigma_X)>0\).  Hence,
with probability tending to one, \(\Sh_{X|j}\) is nonsingular and
\(\eta_j=\Sh_{X|j}^{-1}(\Xb-\Xb_j)\).  Substituting back yields the exact
expansion
\[
\xi_i=(X_i-\Xb_j)^\top\Sh_{X|j}^{-1}(\Xb-\Xb_j),\qquad i:A_i=j,
\]
which proves \eqref{eq:sbw_xiexpansion_car}; in particular the remainder is \(\rho_i\equiv0\).
 
\medskip\noindent\textbf{Step 3 (Preliminary stochastic orders).}
We record the orders needed for \eqref{eq:sbw_xisquare_car}.
 
\emph{(i) Mean difference.}
By Lemma~\ref{lem:car-mean-balance},
$\Xb-\Xb_j=\Opn{n^{-1/2}}.
$
 
\emph{(ii) Covariance and its inverse.}
From \(\Sh_{X|j}\overset{p}{\to}\Sigma_X\), \(\lambda_{\min}(\Sigma_X)>0\), and
continuity of \(\lambda_{\min}\), with probability tending to one \(\Sh_{X|j}\)
is nonsingular and, in the spectral norm,
$
\|\Sh_{X|j}^{-1}\|=\lambda_{\min}(\Sh_{X|j})^{-1}=\Opn{1}.
$
 
\emph{(iii) Second-moment bound.}
For any vector \(v\), \(\|v\|^2=\operatorname{tr}(vv^\top)\), so
$ \frac1{n_j}\sum_{i:A_i=j}\|X_i-\Xb_j\|^2
=\operatorname{tr}(\Sh_{X|j})\overset{p}{\to}\operatorname{tr}(\Sigma_X)<\infty,$
the dimension of \(X_i\) being fixed; hence \(\operatorname{tr}(\Sh_{X|j})=\Opn{1}\),
and since \(n_j/n=\Opn{1}\),
\[
\frac1n\sum_{i:A_i=j}\|X_i-\Xb_j\|^2
=\frac{n_j}{n}\operatorname{tr}(\Sh_{X|j})=\Opn{1}.
\]
With \(\eta_j=\Sh_{X|j}^{-1}(\Xb-\Xb_j)\), parts (i)--(ii) of Step 4 give
\[
\|\eta_j\|\le\|\Sh_{X|j}^{-1}\|\,\|\Xb-\Xb_j\|=\Opn{n^{-1/2}}.
\]
 
\medskip\noindent\textbf{Step 4 (Proof of the second-moment bound in \eqref{eq:sbw_xisquare_car}).}
Using the exact expansion \(\xi_i=(X_i-\Xb_j)^\top\eta_j\),
\[
\frac1n\sum_{i:A_i=j}\xi_i^2
=\frac{n_j}{n}\,\eta_j^\top
\Big
\{\frac1{n_j}\sum_{i:A_i=j}(X_i-\Xb_j)(X_i-\Xb_j)^\top\Big
\}\eta_j
=\frac{n_j}{n}\,\eta_j^\top\Sh_{X|j}\eta_j.
\]
Since \(n_j/n=\Opn{1}\), \(\|\Sh_{X|j}\|=\Opn{1}\), and
\(\|\eta_j\|=\Opn{n^{-1/2}}\),
then $
\frac1n\sum_{i:A_i=j}\xi_i^2
\le\frac{n_j}{n}\,\|\Sh_{X|j}\|\,\|\eta_j\|^2=\Opn{n^{-1}}.$
\end{proof}

\begin{proposition}[EB is balancing-regular under the randomization condition]
\label{prop:eb-car}
For each arm \(j=0,1\), let the EB weights solve
\begin{equation}\label{eq:eb_primal_car}
\min_{\{h_i:A_i=j\}}
\frac1{n_j}\sum_{i:A_i=j}h_i\log h_i
\quad\text{subject to }\eqref{eq:hconstraints-car}.
\end{equation}
Assume randomization condition, that
\eqref{eq:eb_primal_car} admits an interior solution with a finite
dual multiplier \(\lambda_j\), and that
\[
\max_{i:A_i=j}
\left|
\lambda_j^\top(X_i-\Xb_j)
\right|
=
\opn{1},
\qquad
\frac1{n_j}\sum_{i:A_i=j}
\|X_i-\Xb_j\|^4
=
\Opn{1}.
\]
Then the EB weights are balancing-regular. In particular,
\begin{equation}\label{eq:eb_expansion_car}
h_i-1
=
(X_i-\Xb_j)^\top
\Sh_{X|j}^{-1}
(\Xb-\Xb_j)
+\rho_i,
\qquad i:A_i=j,
\end{equation}
where
\begin{equation}\label{eq:eb_remainder_car}
\frac1n\sum_{i:A_i=j}\rho_i^2
=
\opn{n^{-1}},
\qquad
\frac1n\sum_{i:A_i=j}(h_i-1)^2
=
\Opn{n^{-1}}.
\end{equation}
\end{proposition}
 
\begin{proof}
Fix arm \(j\) and write \(\tilde X_i=X_i-\Xb_j\), \(d_j=\Xb-\Xb_j\), and
\(\xi_i=h_i-1\).  Throughout we use \(d_j=\Opn{n^{-1/2}}\), \(\Sh_{X|j}\overset{p}{\to}\Sigma_X\succ0\)
(Lemma~\ref{lem:car-cov-consistency}), and the local EB regularity condition.
 
\medskip\noindent\textbf{Step 1 (Dual/tilting solution).}
Entropy balancing has the exponential-tilting solution
\[
h_i(\lambda_j)
=\frac{\exp(\lambda_j^\top\tilde X_i)}
{\,n_j^{-1}\sum_{k:A_k=j}\exp(\lambda_j^\top\tilde X_k)\,},
\qquad i:A_i=j,
\]
where the dual multiplier \(\lambda_j\) solves the calibration equation
\(\Psi_j(\lambda_j)=d_j\) with
\(\Psi_j(\lambda)=n_j^{-1}\sum_{i:A_i=j}h_i(\lambda)\tilde X_i\)
\citep{hainmueller2012entropy}.
 
\medskip\noindent\textbf{Step 2 (Linearization of the calibration map).}
Since \(\Psi_j(0)=0\) and the Jacobian at the origin is
\(\dot\Psi_j(0)=n_j^{-1}\sum_{i:A_i=j}\tilde X_i\tilde X_i^\top=\Sh_{X|j}\),
the local regularity condition gives the Taylor expansion
\[
\Psi_j(\lambda)=\Sh_{X|j}\lambda+\Opn{\|\lambda\|^2}.
\]
 
\medskip\noindent\textbf{Step 3 (Closed form for \(\lambda_j\)).}
By Lemmas~\ref{lem:car-mean-balance},
\(d_j=\Opn{n^{-1/2}}\) and \(\Sh_{X|j}\overset{p}{\to}\Sigma_X\succ0\), so
\(\|\Sh_{X|j}^{-1}\|=\Opn{1}\).  Inverting the calibration equation
\(\Psi_j(\lambda_j)=d_j\) through the expansion of Step 2,
\[
\lambda_j=\Sh_{X|j}^{-1}d_j+\Opn{n^{-1}}=\Opn{n^{-1/2}}.
\]
 
\medskip\noindent\textbf{Step 4 (Weight expansion and remainder).}
Taylor expanding the tilting weights and using
\(n_j^{-1}\sum_{i:A_i=j}\tilde X_i=0\),
\[
\xi_i=h_i-1=\lambda_j^\top\tilde X_i+\tilde\rho_i,
\qquad
\frac1n\sum_{i:A_i=j}\tilde\rho_i^2=\Opn{n^{-2}}=\opn{n^{-1}},
\]
where \(\tilde\rho_i\) collects the quadratic tilting remainder, controlled by
the local regularity condition.  Substituting the expansion for \(\lambda_j\)
from Step 3 and absorbing the \(\Opn{n^{-1}}\) multiplier error into \(\rho_i\),
\[
\xi_i=(X_i-\Xb_j)^\top\Sh_{X|j}^{-1}(\Xb-\Xb_j)+\rho_i,
\qquad
\frac1n\sum_{i:A_i=j}\rho_i^2=\opn{n^{-1}},
\]
which is the expansion in \eqref{eq:eb_expansion_car}.
 
\medskip\noindent\textbf{Step 5 (Stability bounds).}
The leading term of \(\xi_i\) is \(\eta_j^\top(X_i-\Xb_j)\) with
\(\eta_j=\Sh_{X|j}^{-1}d_j=\Opn{n^{-1/2}}\), exactly as for SBW; the remainder
adds \(n^{-1}\sum_{i:A_i=j}\rho_i^2=\opn{n^{-1}}\). Hence the second-moment bounds
$
\frac1n\sum_{i=1}^n\xi_i^2=\Opn{n^{-1}},
$
follow exactly in Proposition~\ref{prop:sbw-car}.
\end{proof}

\vspace{20pt}
\begin{proposition}[Stabilized IPW is Balancing-regular under Randomization condition]
\label{prop:ipw-car}
Assume the randomization condition. Suppose \(X_i\in\mathbb R^q\) has bounded support. Define
\[
e_i(\gamma)=\expit\{\logit(\pi)+(X_i-\Xb)^\top\gamma\},
\qquad \hat e_i=e_i(\hat\gamma),
\]
where \(\hat\gamma\) solves the centered logistic score equation
\(\sum_{i=1}^n\{A_i-e_i(\hat\gamma)\}(X_i-\Xb)=0\), assumed to have a regular
local solution. 
Define the stabilized IPW weights
\[
h_i=A_i\frac{\pi}{\hat e_i}+(1-A_i)\frac{1-\pi}{1-\hat e_i}.
\]
Then the weights are Balancing-regular: for each \(j=0,1\), write $\xi_i=h_i-1$, we have 
\begin{equation}\label{eq:ipw_xiexpansion_car}
\xi_i=(X_i-\Xb_j)^\top\Sh_{X|j}^{-1}(\Xb-\Xb_j)+\rho_i,
\qquad \frac1n\sum_{i:A_i=j}\rho_i^2=\opn{n^{-1}},
\end{equation}
and
\begin{equation}\label{eq:ipw_xisquare_car}
\frac1n\sum_{i:A_i=j}\xi_i^2=\Opn{n^{-1}}.
\end{equation}
\end{proposition}

\begin{proof}
The first-order expansion of estimated-propensity IPW weights is closely related
to the expansion in \citet{shao2026inverse}.  We recall the argument in the
present notation and verify that.

\medskip\noindent\textbf{Step 1 (Treatment--covariate imbalance).}
Let \(\Delta_n=\frac1n\sum_{i=1}^n(A_i-\pi)(X_i-\Xb)\) as the empirical treatment---covariate imbalance. Since
\(\frac1n\sum_i(X_i-\Xb)=0\),
$ \Delta_n=\frac1n\sum_{i=1}^n A_i(X_i -\Xb)=\frac{n_1}{n}(\Xb_1-\Xb) =-\frac{n_0}{n}(\Xb_0-\Xb),$
which is \(\Opn{n^{-1/2}}\) by Lemma~\ref{lem:car-mean-balance} and
\(n_j/n=\Opn1\). This design-robust bound is the only place the stratified
design enters.

\medskip\noindent\textbf{Step 2 (Expansion of the logistic score and the weights).}
Define the centered logistic score
\[
S_n(\gamma)
=
\frac1n\sum_{i=1}^n
\{A_i-e_i(\gamma)\}(X_i-\Xb),
\]
so that \(\hat\gamma\) satisfies \(S_n(\hat\gamma)=0\). Since
\(e_i(0)=\pi\),
\[
S_n(0)
=
\frac1n\sum_{i=1}^n(A_i-\pi)(X_i-\Xb)
=
\Delta_n .
\]
Moreover,
$\left.
\frac{\partial e_i(\gamma)}{\partial\gamma}
\right|_{\gamma=0}
=
\pi(1-\pi)(X_i-\Xb).$ 
Therefore, a Taylor expansion of \(S_n(\gamma)\) around \(\gamma=0\) gives
\[
0
=
S_n(\hat\gamma) = S_n(0)+\left.
\frac{\partial S_n(\gamma)}{\partial\gamma}
\right|_{\gamma=0}  \hat{\gamma} + R_{\gamma,n}
=
\Delta_n
-
\pi(1-\pi)\hat\Sigma_X\hat\gamma
+
R_{\gamma,n}.
\]
Under bounded support of \(X_i\), the second-order remainder satisfies
$
\|R_{\gamma,n}\|
=
O_p(\|\hat\gamma\|^2).
$
By the assumed regular local solution, we take the local root satisfying
\(\|\hat\gamma\|=\Opn{n^{-1/2}}\). Hence
$
\|R_{\gamma,n}\|
=
O_p(n^{-1})
=
o_p(n^{-1/2}).$
Since \(\hat\Sigma_X\overset{p}{\to}\Sigma_X\succ0\), we have
\(\|\hat\Sigma_X^{-1}\|=\Opn1\). Rearranging the expanded score equation yields
\begin{equation}\label{eq:ipw_gamma_expansion}
\hat\gamma
=
\{\pi(1-\pi)\hat\Sigma_X\}^{-1}\Delta_n
+
o_p(n^{-1/2}).
\end{equation}
This is the centered fixed-intercept analogue of the estimated-propensity
expansion in \citet{shao2026inverse}.

Next, applying the same Taylor expansion to the individual fitted propensity
score \(e_i(\gamma)\), we have, uniformly over \(i\),
\[
\hat e_i-\pi
=
\pi(1-\pi)(X_i-\Xb)^\top\hat\gamma
+
\Opn{\|\hat\gamma\|^2}.
\]
For treated subjects, \(A_i=1\), a first-order Taylor expansion of
\(u\mapsto \pi/u-1\) around \(u=\pi\) gives
\[
\frac{\pi}{\hat e_i}-1
=
-\frac{\hat e_i-\pi}{\pi}
+
O_p\{(\hat e_i-\pi)^2\}
=
-(1-\pi)(X_i-\Xb)^\top\hat\gamma
+
\Opn{\|\hat\gamma\|^2}.
\]
For control subjects, \(A_i=0\), a first-order Taylor expansion of
\(u\mapsto (1-\pi)/(1-u)-1\) around \(u=\pi\) gives
\[
\frac{1-\pi}{1-\hat e_i}-1
=
\frac{\hat e_i-\pi}{1-\pi}
+
O_p\{(\hat e_i-\pi)^2\}
=
\pi(X_i-\Xb)^\top\hat\gamma
+
\Opn{\|\hat\gamma\|^2}.
\]
Combining the two cases,
$
\xi_i = h_i-1
=
-(A_i-\pi)(X_i-\Xb)^\top\hat\gamma
+
\Opn{\|\hat\gamma\|^2},
$
uniformly over \(i\). Since \(\hat\gamma=\Opn{n^{-1/2}}\), the quadratic
remainder has squared average \(\Opn{n^{-2}}=\opn{n^{-1}}\). Substituting
\eqref{eq:ipw_gamma_expansion} yields
\begin{equation}\label{eq:ipw_pooled_expansion}
\xi_i
=
-\frac{A_i-\pi}{\pi(1-\pi)}
(X_i-\Xb)^\top\hat\Sigma_X^{-1}\Delta_n
+
r_i^{*},
\qquad
\frac1n\sum_{i=1}^n(r_i^{*})^2=\opn{n^{-1}}.
\end{equation}
This is the centered fixed-intercept version of the estimated-propensity
weight expansion in \citet{shao2026inverse}.
\medskip\noindent\textbf{Step 3 (Arm-specific calibration form).}
For \(A_i=1\), substitute \(\Delta_n=\frac{n_1}{n}(\Xb_1-\Xb)\) into
\eqref{eq:ipw_pooled_expansion}:
\[
\xi_i=\frac{n_1/n}{\pi}(X_i-\Xb)^\top\hat\Sigma_X^{-1}(\Xb-\Xb_1)+r_i^{*}.
\]
We replace this leading term by the target
\((X_i-\Xb_1)^\top\Sh_{X|1}^{-1}(\Xb-\Xb_1)\) through three substitutions, each
changing it by an \(L_2\)-negligible amount: the coefficient \(\frac{n_1/n}{\pi}\)
by \(1\); the centering \(X_i-\Xb\) by \(X_i-\Xb_1\); and the pooled covariance
\(\hat\Sigma_X^{-1}\) by the arm-specific \(\Sh_{X|1}^{-1}\). The three induced
differences are
\[
\Big(\tfrac{n_1/n}{\pi}-1\Big)(X_i-\Xb)^\top\hat\Sigma_X^{-1}(\Xb-\Xb_1),
\quad
(\Xb_1-\Xb)^\top\hat\Sigma_X^{-1}(\Xb-\Xb_1),
\quad
(X_i-\Xb_1)^\top(\hat\Sigma_X^{-1}-\Sh_{X|1}^{-1})(\Xb-\Xb_1).
\]
Each is a product of factors that are \(\Opn{n^{-1/2}}\) or \(\opn1\) (the
coefficient gap and \(\Xb-\Xb_1\) by Lemma~\ref{lem:car-mean-balance};
\(\hat\Sigma_X^{-1}-\Sh_{X|1}^{-1}=\opn1\) by Lemma~\ref{lem:car-cov-consistency})
and factors that are \(\Opn1\) (the bounded covariates and
\(\|\hat\Sigma_X^{-1}\|\)); together with \(\Xb-\Xb_1=\Opn{n^{-1/2}}\), each
difference has squared average \(\opn{n^{-1}}\). Collecting the three differences
and \(r_i^{*}\) into a single remainder \(\rho_i\),
\[
\xi_i=(X_i-\Xb_1)^\top\Sh_{X|1}^{-1}(\Xb-\Xb_1)+\rho_i,
\qquad \frac1n\sum_{i:A_i=1}\rho_i^2=\opn{n^{-1}}.
\]
The control arm, using \(\Delta_n=\frac{n_0}{n}(\Xb-\Xb_0)\) and replacing
\(\pi,\Xb_1\) by \(1-\pi,\Xb_0\), is identical. This proves
\eqref{eq:ipw_xiexpansion_car}.

\medskip\noindent\textbf{Step 4 (Stability bounds).}
The bounds \eqref{eq:ipw_xisquare_car} follow exactly as in
Proposition~\ref{prop:sbw-car}: with \(\eta_j=\Sh_{X|j}^{-1}(\Xb-\Xb_j)\),
\(\|\eta_j\|=\Opn{n^{-1/2}}\) and
\(\frac1n\sum_{i:A_i=j}\|X_i-\Xb_j\|^2=\Opn1\)
(Lemmas~\ref{lem:car-mean-balance}) give
$\frac1n\sum_{i:A_i=j}\xi_i^2=\Opn{n^{-1}}.$
\end{proof}

\subsection{Step 4 -- Reduction to Covariate-Adjusted Score}
\label{subsec:car-step4}

\begin{theorem}[Log-Rank Score Equivalence]
\label{the:weighted-ye-bridge-car}
Assume the weights are Balancing-regular.
Then
\begin{equation}\label{eq:Uh_linear_expansion_main}
  \hat U_h
  =\hat U_\mathrm{unadj}
  +\frac{1}{n}\sum_{i=1}^n
   \bigl\{A_i\xi_i\hat O_{i1}-(1-A_i)\xi_i\hat O_{i0}\bigr\}
  +\opn{n^{-1/2}},
\end{equation}
and consequently
\begin{equation}\label{eq:weighted_ye_bridge_car}
\Uh_h=\Uh_\mathrm{aug}+\opn{n^{-1/2}}.
\end{equation}
\end{theorem}
 
\begin{proof}
By Lemma~\ref{lem:score_bridge},
\[
\Uh_h-\Uh_\mathrm{unadj}
=\frac1n\sum_{i:A_i=1}\xi_i\Oh_{i1}-\frac1n\sum_{i:A_i=0}\xi_i\Oh_{i0}
+\opn{n^{-1/2}},
\]
For each arm \(j\), substitute the calibration expansion
\(\xi_i=(X_i-\Xb_j)^\top\Sh_{X|j}^{-1}(\Xb-\Xb_j)+\rho_i\). The linear term is an exact identity and the remainder term is $\opn{n^{-1/2}}$ because, by
Cauchy--Schwarz, the $L_2$-regular remainder
$\frac1n\sum_{i:A_i=j}\rho_i^2=\opn{n^{-1}}$ and the boundedness of $\Oh_{ij}$
give
\[
\Big|\frac1n\sum_{i:A_i=j}\rho_i\Oh_{ij}\Big|
\le\Big(\frac1n\sum_{i:A_i=j}\rho_i^2\Big)^{1/2}
   \Big(\frac1n\sum_{i:A_i=j}\Oh_{ij}^2\Big)^{1/2}
=\opn{n^{-1/2}}\cdot\Opn1=\opn{n^{-1/2}}.
\]
Hence, 
\[
\frac1n\sum_{i:A_i=j}\xi_i\Oh_{ij}
=(\Xb-\Xb_j)^\top\Sh_{X|j}^{-1}\,
 \underbrace{\frac1n\sum_{i:A_i=j}(X_i-\Xb_j)\Oh_{ij}}_{=\,(n_j/n)\,\Sh_{X|j}\bh_j}
+\frac1n\sum_{i:A_i=j}\rho_i\Oh_{ij}
=\frac{n_j}{n}(\Xb-\Xb_j)^\top\bh_j+\opn{n^{-1/2}},
\]
using \(\Sh_{X|j}^{-1}\Sh_{X|j}=\mathrm{I}\) in the last step.  Therefore
\[
\Uh_h-\Uh_\mathrm{unadj}
=\frac{n_1}{n}(\Xb-\Xb_1)^\top\bh_1-\frac{n_0}{n}(\Xb-\Xb_0)^\top\bh_0
+\opn{n^{-1/2}}.
\]
Using the centering identities
\(n^{-1}\sum_iA_i(X_i-\Xb)=(n_1/n)(\Xb_1-\Xb)\) and
\(n^{-1}\sum_i(1-A_i)(X_i-\Xb)=(n_0/n)(\Xb_0-\Xb)\), together with the sign
reversal \(\Xb-\Xb_j=-(\Xb_j-\Xb)\),
\[
\frac{n_1}{n}(\Xb-\Xb_1)^\top\bh_1
=-\frac1n\sum_i A_i(X_i-\Xb)^\top\bh_1,
\qquad
-\frac{n_0}{n}(\Xb-\Xb_0)^\top\bh_0
=\frac1n\sum_i (1-A_i)(X_i-\Xb)^\top\bh_0,
\]
so that
\[
\Uh_h-\Uh_\mathrm{unadj}
=-\frac1n\sum_{i=1}^n\bigl\{A_i(X_i-\Xb)^\top\bh_1-(1-A_i)(X_i-\Xb)^\top\bh_0\bigr\}
+\opn{n^{-1/2}}
=\Uh_\mathrm{aug}-\Uh_\mathrm{unadj}+\opn{n^{-1/2}}.
\]
\end{proof}

\section{Equivalence for the marginal Cox Hazard Ratio}
\label{subsec:q3-cox-hr}

\paragraph{Roadmap.}
\begin{description}
  \item[Step 1 (Lemma~\ref{lem:cox_identity}).]
  Establish an exact finite-sample representation of the weighted Cox score in terms of weighted Cox derived outcomes.
  \item[Step 2 (Lemma~\ref{lem:cox_score_bridge}).]
  Linearize the perturbed Cox risk-set ratio uniformly over
  \(\Theta_0\) and obtain the first-order score expansion
  $
  \Uh_h(\vartheta)-\Uh_\mathrm{unadj}(\vartheta)
  =
  \frac1n\sum_{i=1}^n
  \xi_i
  \{A_i\Oh_{i1}(\vartheta)
  -(1-A_i)\Oh_{i0}(\vartheta)\}
  +\opn{n^{-1/2}}.$
  \item[Step 3 (Proposition~\ref{prop:cox_bal_bridge}).]
  Use balancing regularity to reduce the leading weighted-score
  perturbation to Ye's arm-specific linear augmentation. By
  Propositions~\ref{prop:sbw-car}, \ref{prop:eb-car}, and
  \ref{prop:ipw-car}, this reduction applies to SBW, EB, and stabilized
  IPW.

  \item[Step 4 (Lemma~\ref{lem:cox_consistency_freeze} and
  Theorem~\ref{thm:cox_root_equiv}).]
  Establish \(n^{-1/2}\)-consistency, freeze the augmentation at
  \(\hat\theta_\mathrm{unadj}\), and prove root equivalence.
\end{description}

The weight expansions established in the preceding section are outcome-free. Hence no additional weight-specific argument is required for the Cox endpoint: once a weighting scheme is balancing-regular, the generic Cox results below apply. Following \cite{lin2000fitting, binder1992fitting}, We now introduce a weighted Cox score.  At this point, the weights \(h_i\) are
arbitrary nonnegative analysis weights.  Define
\[
D_\vartheta^{(h)}(t)
=
e^\vartheta\Yb_1^{(h)}(t)+\Yb_0^{(h)}(t),
\qquad
R_\vartheta^{(h)}(t)
=
\frac{e^\vartheta\Yb_1^{(h)}(t)}{D_\vartheta^{(h)}(t)}.
\]
The weighted marginal Cox score is
\begin{equation}\label{eq:Uhtheta}
\Uh_h(\vartheta)
=
\frac1n\sum_{i=1}^n
h_i
\int_0^\tau
\{A_i-R_\vartheta^{(h)}(t)\}\,dN_i(t)
=
\int_0^\tau
\frac{\Yb_0^{(h)}(t)}{D_\vartheta^{(h)}(t)}
\,d\Nb_1^{(h)}(t)
-
\int_0^\tau
\frac{e^\vartheta\Yb_1^{(h)}(t)}{D_\vartheta^{(h)}(t)}
\,d\Nb_0^{(h)}(t).
\end{equation}
The derivation process above follows the  algebra we used in Lemma \ref{lem:identity}. The same weights \(h_i\) enter both the event contribution and the risk-set
denominator.  When \(h_i\equiv1\), \(\Uh_h(\vartheta)\) reduces to the
unweighted Cox score \(\Uh_\mathrm{unadj}(\vartheta)\).  When additionally
\(\vartheta=0\), it reduces to the unadjusted log-rank numerator.

\subsection{Step 1 -- Weighted Cox algebraic identity}
\label{subsec:cox-step1}

\begin{lemma}[Cox Score derived-outcome identity] \label{lem:cox_identity}
For any nonnegative weights with \(D_\vartheta^{(h)}(t)>0\) on \([0,\tau]\),
\begin{equation}\label{eq:cox_identity}
\Uh_h(\vartheta)
\equiv
\frac1n\sum_{i=1}^n
h_i
\left\{
A_i\Oh_{i1}^{(h)}(\vartheta)
-
(1-A_i)\Oh_{i0}^{(h)}(\vartheta)
\right\}.
\end{equation}
\end{lemma}
 
\begin{proof}
The argument follows two-stage cancellation as in
Lemma~\ref{lem:identity}.

\medskip\noindent\textbf{Stage A.}
From \(\Uh_h(\vartheta)=\frac1n\sum_ih_i\int_0^\tau\{A_i-R_\vartheta^{(h)}\}dN_i\),
the same distributivity step as in Lemma~\ref{lem:identity} gives
\(\Uh_h(\vartheta)=\int_0^\tau d\Nb_1^{(h)}-\int_0^\tau R_\vartheta^{(h)}\,d\Nb^{(h)}\).
Splitting \(d\Nb^{(h)}=d\Nb_1^{(h)}+d\Nb_0^{(h)}\) and using
\(1-R_\vartheta^{(h)}=\Yb_0^{(h)}/D_\vartheta^{(h)}\),
\(R_\vartheta^{(h)}=e^\vartheta\Yb_1^{(h)}/D_\vartheta^{(h)}\), we have the last term of \ref{eq:Uhtheta}.

\medskip\noindent\textbf{Stage B.}
The ratios in \(\Oh_{i1}^{(h)}(\vartheta)\) do not depend on \(i\); pulling them
out and using \(\frac1n\sum_{i:A_i=1}h_i\,dN_{i1}=d\Nb_1^{(h)}\),
\(\frac1n\sum_{i:A_i=1}h_iY_{i1}=\Yb_1^{(h)}\),
\begin{equation}\label{eq:cox_id_treated}
\frac1n\sum_i h_iA_i\Oh_{i1}^{(h)}(\vartheta)
=\int_0^\tau\frac{\Yb_0^{(h)}}{D_\vartheta^{(h)}}\,d\Nb_1^{(h)}
-\int_0^\tau\frac{e^\vartheta\Yb_1^{(h)}\Yb_0^{(h)}}{\{D_\vartheta^{(h)}\}^2}\,d\Nb^{(h)}.
\end{equation}
By the same pull-out argument for the control arm,
\begin{equation}\label{eq:cox_id_control}
\frac1n\sum_i h_i(1-A_i)\Oh_{i0}^{(h)}(\vartheta)
=\int_0^\tau\frac{e^\vartheta\Yb_1^{(h)}}{D_\vartheta^{(h)}}\,d\Nb_0^{(h)}
-\int_0^\tau\frac{e^\vartheta\Yb_1^{(h)}\Yb_0^{(h)}}{\{D_\vartheta^{(h)}\}^2}\,d\Nb^{(h)}.
\end{equation}
Subtracting \eqref{eq:cox_id_control} from \eqref{eq:cox_id_treated}, the two
double integrals cancel exactly, leaving the same as Stage A results.
\end{proof}
 
\subsection{Step 2 -- Uniform linearization and score-perturbation representation}
\label{subsec:cox-step2}
\begin{lemma}[Cox derived-outcome representation of the score perturbation]
\label{lem:cox_score_bridge}
Let \(\Theta_0=[\theta_0-a,\theta_0+a]\) be a fixed compact
neighborhood of \(\theta_0\), and define
$
\bar c=e^{\theta_0+a},
\kappa=\min(e^{\theta_0-a},1)>0.$ Assume for \(j=0,1\),
\begin{equation}\label{eq:cox_arm_l2}
\frac1n\sum_{i:A_i=j}\xi_i^2
=
\Opn{n^{-1}}.
\end{equation}
Then, uniformly over \(\vartheta\in\Theta_0\),
\begin{equation}\label{eq:cox_score_bridge}
\Uh_h(\vartheta)-\Uh_\mathrm{unadj}(\vartheta)
=
\frac1n\sum_{i=1}^n
\xi_i
\bigl\{
A_i\Oh_{i1}(\vartheta)
-(1-A_i)\Oh_{i0}(\vartheta)
\bigr\}
+\opn{n^{-1/2}}.
\end{equation}
\end{lemma}
\begin{proof}
Part~I linearizes the weighted risk-set ratio
\(R_\vartheta^{(h)}\) around \(R_\vartheta\), uniformly over \(\vartheta\in\Theta_0\)
and \(t\in[0,\tau]\); this is the Cox analogue of Lemma~\ref{lem:ratio}, the only
new ingredient being that the factor \(e^\vartheta\) must be controlled uniformly
on the compact set \(\Theta_0\).  Part~II inserts this expansion into the score
difference and identifies the result with the Cox derived outcomes, as in
Lemma~\ref{lem:score_bridge}. Throughout write
\[
\Delta_j(t)=\Yb_j^{(h)}(t)-\Yb_j(t)=\frac1n\sum_{i:A_i=j}\xi_iY_i(t),
\qquad j=0,1,
\]
and \(\Delta_D(t,\vartheta)=D_\vartheta^{(h)}(t)-D_\vartheta(t)
=e^\vartheta\Delta_1(t)+\Delta_0(t)\).
 
\medskip
\noindent\textbf{Part I: ratio linearization.}
We show that, uniformly over \(\vartheta\in\Theta_0\) and \(t\in[0,\tau]\),
\begin{equation}\label{eq:cox_ratio_expansion}
R_\vartheta^{(h)}(t)-R_\vartheta(t)
=
\frac1{D_\vartheta(t)}
\frac1n\sum_{i=1}^n
\xi_i
\Bigl[
e^\vartheta A_i\{1-R_\vartheta(t)\}
-
(1-A_i)R_\vartheta(t)
\Bigr]Y_i(t)
+
r_n(t,\vartheta),
\end{equation}
with \(\sup_{\vartheta\in\Theta_0, t\le\tau}|r_n(t,\vartheta)|=\opn{n^{-1/2}}\);
in particular
\(\sup_{\vartheta\in\Theta_0, t\le\tau}|R_\vartheta^{(h)}(t)-R_\vartheta(t)|
=\Opn{n^{-1/2}}\).
 
\emph{(I.a) Perturbation bounds.}
By Cauchy--Schwarz and \(0\le Y_i(t)\le1\), for each
\(j=0,1\),
\[
\begin{aligned}
\sup_{t\le\tau}|\Delta_j(t)|
=
\sup_{t\le\tau}
\Big|
\frac1n\sum_{i:A_i=j}\xi_iY_i(t)
\Big|
\le
\Big\{
\frac1n\sum_{i:A_i=j}\xi_i^2
\Big\}^{1/2}
\sup_{t\le\tau}
\Big\{
\frac1n\sum_{i:A_i=j}Y_i(t)^2
\Big\}^{1/2}\le
\Big\{
\frac1n\sum_{i:A_i=j}\xi_i^2
\Big\}^{1/2}
=
\Opn{n^{-1/2}}.
\end{aligned}
\]
Since \(e^\vartheta\le\bar c\) uniformly over
\(\vartheta\in\Theta_0\),
\begin{equation}\label{eq:cox_DeltaD_bound}
\sup_{\vartheta\in\Theta_0, t\le\tau}
|\Delta_D(t,\vartheta)|
\le
\bar c\times \sup_{t\le\tau}|\Delta_1(t)|
+
\sup_{t\le\tau}|\Delta_0(t)|
=
\Opn{n^{-1/2}}.
\end{equation}
 
\emph{(I.b) Exact algebraic decomposition.}
Using \(R_\vartheta=e^\vartheta\Yb_1/D_\vartheta\),
\[
R_\vartheta^{(h)}(t)-R_\vartheta(t)
=\frac{e^\vartheta\Yb_1^{(h)}(t)D_\vartheta(t)-e^\vartheta\Yb_1(t) D_\vartheta^{(h)}(t)}
       {D_\vartheta^{(h)}(t)D_\vartheta(t)}
=\frac{e^\vartheta\Delta_1(t)-R_\vartheta(t)\Delta_D(t,\vartheta)}
       {D_\vartheta^{(h)}(t)} .
\]
Writing \(1/D_\vartheta^{(h)}=1/D_\vartheta-\Delta_D/(D_\vartheta D_\vartheta^{(h)})\),
\begin{equation}\label{eq:cox_ratio_decomp}
R_\vartheta^{(h)}(t)-R_\vartheta(t)
= \frac{e^\vartheta\Delta_1(t)-R_\vartheta(t)\Delta_D(t,\vartheta)}{D_\vartheta(t)}
-\frac{\{e^\vartheta\Delta_1(t)-R_\vartheta(t)\Delta_D(t,\vartheta)\}\Delta_D(t,\vartheta)}
   {D_\vartheta(t)D_\vartheta^{(h)}(t)} :=\mathcal A_\vartheta(t) - \mathcal B_\vartheta(t) .
\end{equation}
Since \(\Delta_1=\tfrac1n\sum_i\xi_iA_iY_i\) and
\(\Delta_0=\tfrac1n\sum_i\xi_i(1-A_i)Y_i\),
\[
e^\vartheta\Delta_1-R_\vartheta\Delta_D
=e^\vartheta(1-R_\vartheta)\Delta_1-R_\vartheta\Delta_0
=\frac1n\sum_{i=1}^n\xi_i\bigl[e^\vartheta A_i(1-R_\vartheta)-(1-A_i)R_\vartheta\bigr]Y_i,
\]
so \(\mathcal A_\vartheta(t)\) is exactly the linear term in
\eqref{eq:cox_ratio_expansion} and \(r_n(t,\vartheta)=-\mathcal B_\vartheta(t)\).
 
\emph{(I.c) Uniform bound for \(\mathcal B_\vartheta\).}
For every \(\vartheta\in\Theta_0\),
\[
D_\vartheta(t)=e^\vartheta\Yb_1(t)+\Yb_0(t)\ge\kappa\,\Yb(t),
\qquad
D_\vartheta^{(h)}(t)\ge\kappa\,\Yb^{(h)}(t),
\qquad \kappa=\min(e^{\theta_0-a},1)>0.
\]
The positivity established in the proof of Lemma~\ref{lem:ratio} gives
\(\inf_{t\le\tau}\Yb(t)\ge c/2\) and \(\inf_{t\le\tau}\Yb^{(h)}(t)\ge c/4\) with
probability tending to one; hence
\(\inf_{\vartheta,t}D_\vartheta\ge\kappa c/2\),
\(\inf_{\vartheta,t}D_\vartheta^{(h)}\ge\kappa c/4\), so that
\begin{equation}\label{eq:cox_denomOp1}
\sup_{\vartheta\in\Theta_0, t\le\tau}
\frac1{D_\vartheta(t)D_\vartheta^{(h)}(t)}
\le\frac{8}{\kappa^2 c^2}=\Opn{1}.
\end{equation}
Since \(0\le R_\vartheta\le1\), by \eqref{eq:cox_DeltaD_bound} and (I.a),
\[
\sup_{\vartheta,t} \bigl |e^\vartheta\Delta_1- R_\vartheta \Delta_D\bigr|
\le\bar c\times \sup_{t\le\tau}|\Delta_1|+\sup_{\vartheta,t}|\Delta_D|=\Opn{n^{-1/2}},
\]
and multiplying by \(\sup_{\vartheta,t}|\Delta_D|=\Opn{n^{-1/2}}\) gives
\(\sup_{\vartheta,t}|\{e^\vartheta\Delta_1-R_\vartheta\Delta_D\}\Delta_D|
=\Opn{n^{-1}}\).  Combining with \eqref{eq:cox_denomOp1},
we have $\sup_{\vartheta\in\Theta_0, t\le\tau}|\mathcal B_\vartheta(t)|
\le\Opn{1}\cdot\Opn{n^{-1}}=\opn{n^{-1/2}},
$
which is the remainder bound; \(\mathcal A_\vartheta\) has
numerator \(\Opn{n^{-1/2}}\) and denominator bounded below, giving
\(\sup_{\vartheta,t}|R_\vartheta^{(h)}-R_\vartheta|=\Opn{n^{-1/2}}\). 

\medskip
\noindent\textbf{Part II: score difference and derived-outcome identification.}
From \eqref{eq:Uhtheta} and definition of $\hat{U}_L(\vartheta)$,
\[
\Uh_h(\vartheta)-\Uh_\mathrm{unadj}(\vartheta)
=\frac1n\sum_{i=1}^n
\int_0^\tau
\Bigl[h_i\{A_i-R_\vartheta^{(h)}(t)\}-\{A_i-R_\vartheta(t)\}\Bigr]\,dN_i(t).
\]
The integrand decomposes exactly as
\(h_i\{A_i-R_\vartheta^{(h)}\}-\{A_i-R_\vartheta\}
=\xi_i\{A_i-R_\vartheta\}-h_i\{R_\vartheta^{(h)}-R_\vartheta\}\), so
\[
\Uh_h(\vartheta)-\Uh_\mathrm{unadj}(\vartheta)
=\frac1n\sum_i\xi_i\int_0^\tau\{A_i-R_\vartheta\}\,dN_i
-\frac1n\sum_i h_i\int_0^\tau\{R_\vartheta^{(h)}-R_\vartheta\}\,dN_i
   := \mathcal A_n - \mathcal B_n^{(h)}  .
\]
Split \(\mathcal B_n^{(h)}\) into an unweighted part and a \(\xi\)-weighted part.
By Part~I, \(\sup_{\vartheta,t}|R_\vartheta^{(h)}-R_\vartheta|=\Opn{n^{-1/2}}\); and
by \eqref{eq:cox_arm_l2} with \(N_i(\tau)\le1\) and Cauchy--Schwarz,
\(\frac1n\sum_i|\xi_i|N_i(\tau)\le(\frac1n\sum_i\xi_i^2)^{1/2}=\Opn{n^{-1/2}}\).
Hence the \(\xi\)-weighted part is, uniformly over \(\vartheta\in\Theta_0\),
\[
\Bigl|\frac1n\sum_i\xi_i\int_0^\tau\{R_\vartheta^{(h)}-R_\vartheta\}\,dN_i\Bigr|
\le\sup_{\vartheta,t}|R_\vartheta^{(h)}-R_\vartheta|\cdot\frac1n\sum_i|\xi_i|N_i(\tau)
=\Opn{n^{-1}}=\opn{n^{-1/2}},
\]
so that $
\Uh_h(\vartheta)-\Uh_\mathrm{unadj}(\vartheta)
=\mathcal A_n-\int_0^\tau\{R_\vartheta^{(h)}(t)-R_\vartheta(t)\}\,d\Nb(t)+\opn{n^{-1/2}} .$
Insert \eqref{eq:cox_ratio_expansion}.  Since \(\Nb(\tau)\le1\) and
\(\sup_{\vartheta,t}|r_n|=\opn{n^{-1/2}}\), the integrated remainder is
\(\opn{n^{-1/2}}\) uniformly in \(\vartheta\), so
\[
\int_0^\tau\{R_\vartheta^{(h)}-R_\vartheta\}\,d\Nb
=\frac1n\sum_{i=1}^n\xi_i
\int_0^\tau
\frac{\bigl[e^\vartheta A_i\{1-R_\vartheta\}-(1-A_i)R_\vartheta\bigr]Y_i(t)}
     {D_\vartheta(t)}\,d\Nb(t)
+\opn{n^{-1/2}} .
\]
Combining the last two displays,
\[
\Uh_h(\vartheta)-\Uh_\mathrm{unadj}(\vartheta)
=\frac1n\sum_{i=1}^n\xi_i
\Biggl[
\int_0^\tau\{A_i-R_\vartheta\}\,dN_i
-\int_0^\tau
\frac{\bigl[e^\vartheta A_i\{1-R_\vartheta\}-(1-A_i)R_\vartheta\bigr]Y_i}
     {D_\vartheta}\,d\Nb
\Biggr]+\opn{n^{-1/2}} .
\]
Next we identify the bracket, if \(A_i=1\), then \(N_i=N_{i1}\),
\(Y_i=Y_{i1}\), and \(1-R_\vartheta=\Yb_0/D_\vartheta\), so the bracket equals
$
\int_0^\tau\frac{\Yb_0(t)}{D_\vartheta(t)}
\Bigl\{dN_{i1}(t)-Y_{i1}(t)\frac{e^\vartheta\,d\Nb(t)}{D_\vartheta(t)}\Bigr\}
=\Oh_{i1}(\vartheta).$
If \(A_i=0\), then \(N_i=N_{i0}\), \(Y_i=Y_{i0}\), and
\(R_\vartheta=e^\vartheta\Yb_1/D_\vartheta\), so the bracket equals
\(-\Oh_{i0}(\vartheta)\).
\end{proof}

The balancing-regular expansion is outcome-free. To apply it to the
Cox score, we verify that the Cox derived outcomes and their associated
arm-specific regression slopes are uniformly well behaved over a
neighborhood of the target parameter.

\begin{lemma}[Uniform regularity of the Cox derived outcomes]
\label{lem:cox_Otheta_bounded}
Assume the at-risk positivity condition and the covariance conclusion
of Lemma~\ref{lem:car-cov-consistency}. Let
\(\Theta_0=[\theta_0-a,\theta_0+a]\) be compact. Then there exist finite
deterministic constants \(C_O\) and \(C_{\dot O}\) such that, with
probability tending to one,
\[
\sup_{\vartheta\in\Theta_0}\max_{1\le i\le n}\max_{j=0,1}
|\Oh_{ij}(\vartheta)|\le C_O,
\qquad
\sup_{\vartheta\in\Theta_0}\max_{1\le i\le n}\max_{j=0,1}
|\partial_\vartheta\Oh_{ij}(\vartheta)|\le C_{\dot O}.
\]
Consequently, for \(j=0,1\),
\[
\sup_{\vartheta\in\Theta_0}
\frac1n\sum_{i:A_i=j}\Oh_{ij}(\vartheta)^2=\Opn{1},
\qquad
\sup_{\vartheta\in\Theta_0}\|\bh_j(\vartheta)\|=\Opn{1},
\qquad
\sup_{\vartheta\in\Theta_0}
\|\partial_\vartheta\bh_j(\vartheta)\|=\Opn{1}.
\]
\end{lemma}

\begin{proof}
Let
\(\kappa=\min\{e^{\theta_0-a},1\}>0\) and
\(\mathcal E_n=\{\inf_{t\le\tau}\Yb(t)\ge c/2\}\).
By the positivity argument in Lemma~\ref{lem:ratio},
\(\Pr(\mathcal E_n)\to1\). On \(\mathcal E_n\), uniformly over
\((\vartheta,t)\in\Theta_0\times[0,\tau]\),
\[
D_\vartheta(t)
=e^\vartheta\Yb_1(t)+\Yb_0(t)
\ge\kappa\Yb(t)\ge\kappa c/2=:d_*>0.
\]

Write
\[
p_\vartheta=\frac{e^\vartheta\Yb_1}{D_\vartheta},
\qquad
q_\vartheta=\frac{\Yb_0}{D_\vartheta}=1-p_\vartheta,
\qquad
g_{1,\vartheta}=\frac{e^\vartheta\Yb_0}{D_\vartheta^2},
\qquad
g_{0,\vartheta}=\frac{e^\vartheta\Yb_1}{D_\vartheta^2}.
\]
On \(\mathcal E_n\), \(0\le p_\vartheta,q_\vartheta\le1\) and
\(0\le g_{j,\vartheta}\le G_\Theta\), where
\(G_\Theta=e^{\theta_0+a}/d_*^2<\infty\). Moreover,
\[
\Oh_{i1}(\vartheta)
=
\int_0^\tau q_\vartheta(t)\,dN_{i1}(t)
-
\int_0^\tau Y_{i1}(t)g_{1,\vartheta}(t)\,d\Nb(t), \quad
\Oh_{i0}(\vartheta)
=
\int_0^\tau p_\vartheta(t)\,dN_{i0}(t)
-
\int_0^\tau Y_{i0}(t)g_{0,\vartheta}(t)\,d\Nb(t).
\]
Since \(\int_0^\tau dN_{ij}\le1\), \(\Nb(\tau)\le1\), and
\(0\le Y_{ij}\le1\), it follows that
\[
\sup_{\vartheta\in\Theta_0}\max_{i,j}
|\Oh_{ij}(\vartheta)|
\le1+G_\Theta.
\]

For the derivatives,
\[
\partial_\vartheta p_\vartheta=p_\vartheta q_\vartheta,
\qquad
\partial_\vartheta q_\vartheta=-p_\vartheta q_\vartheta,
\qquad
\partial_\vartheta g_{j,\vartheta}
=g_{j,\vartheta}(1-2p_\vartheta).
\]
Thus
\(|\partial_\vartheta p_\vartheta|,
 |\partial_\vartheta q_\vartheta|\le1/4\) and
\(|\partial_\vartheta g_{j,\vartheta}|\le G_\Theta\). Differentiating
the two displays above therefore gives
\[
\sup_{\vartheta\in\Theta_0}\max_{i,j}
|\partial_\vartheta\Oh_{ij}(\vartheta)|
\le\frac14+G_\Theta.
\]
Hence one may take \(C_O=1+G_\Theta\) and
\(C_{\dot O}=1/4+G_\Theta\). The first consequence follows immediately from
\[
\sup_{\vartheta\in\Theta_0}
\frac1n\sum_{i:A_i=j}\Oh_{ij}(\vartheta)^2
\le\frac{n_j}{n}C_O^2=\Opn{1}.
\]
Furthermore, by Cauchy--Schwarz,
\[
\begin{aligned}
\sup_{\vartheta\in\Theta_0}\|\bh_j(\vartheta)\|
\le
\|\Sh_{X|j}^{-1}\|
\big\{
\frac1{n_j}\sum_{i:A_i=j}\|X_i-\Xb_j\|^2
\big\}^{1/2}C_O  
=
\|\Sh_{X|j}^{-1}\|
\{\operatorname{tr}(\Sh_{X|j})\}^{1/2}C_O
=\Opn{1},
\end{aligned}
\]
where Lemma~\ref{lem:car-cov-consistency} gives
\(\|\Sh_{X|j}^{-1}\|=\Opn{1}\) and
\(\operatorname{tr}(\Sh_{X|j})=\Opn{1}\). Finally, because \(\Sh_{X|j}\), \(X_i\), and \(\Xb_j\) do not depend on
\(\vartheta\),
\[
\partial_\vartheta\bh_j(\vartheta)
=
\Sh_{X|j}^{-1}
\frac1{n_j}\sum_{i:A_i=j}
(X_i-\Xb_j)\partial_\vartheta\Oh_{ij}(\vartheta).
\]
The same Cauchy--Schwarz bound, with \(C_{\dot O}\) in place of \(C_O\),
yields
\(\sup_{\vartheta\in\Theta_0}
\|\partial_\vartheta\bh_j(\vartheta)\|=\Opn{1}\).
\end{proof}

\subsection{Step 3 -- Reduction to Covariate-adjusted Cox score}
\label{subsec:cox-step4}

\begin{proposition}[Balancing-regular reduction of the Cox score perturbation]
\label{prop:cox_bal_bridge}
Assume the conditions of Lemma~\ref{lem:cox_score_bridge} and suppose
that the weights are balancing-regular. Then, uniformly over
$\vartheta\in\Theta_0$,
\begin{equation}\label{eq:cox_cal_reduction}
\Uh_h(\vartheta)
=\Uh_\mathrm{unadj}(\vartheta)
-\frac1n\sum_{i=1}^n
\Bigl\{
A_i(X_i-\Xb)^\top\bh_1(\vartheta)
-(1-A_i)(X_i-\Xb)^\top\bh_0(\vartheta)
\Bigr\}
+\opn{n^{-1/2}}.
\end{equation}
\end{proposition}

\begin{proof}
For each arm \(j\), balancing regularity gives
$
\xi_i
=
(X_i-\Xb_j)^\top
\Sh_{X|j}^{-1}(\Xb-\Xb_j)
+\rho_i,$ with $
\frac1n\sum_{i:A_i=j}\rho_i^2
=
\opn{n^{-1}}.$
Moreover, by Cauchy--Schwarz and
Lemma~\ref{lem:cox_Otheta_bounded},
\[
\sup_{\vartheta\in\Theta_0}
\Big|
\frac1n\sum_{i:A_i=j}
\rho_i\Oh_{ij}(\vartheta)
\Big|
\le
\Big\{
\frac1n\sum_{i:A_i=j}\rho_i^2
\Big\}^{1/2}
\Big\{
\frac1n\sum_{i:A_i=j}
\sup_{\vartheta\in\Theta_0}
\Oh_{ij}(\vartheta)^2
\Big\}^{1/2}
=
\opn{n^{-1/2}}.
\]

By Lemma~\ref{lem:cox_score_bridge}, uniformly over $\vartheta\in\Theta_0$,
$
\Uh_h(\vartheta)-\Uh_\mathrm{unadj}(\vartheta)
=\frac1n\sum_{i:A_i=1}\xi_i\Oh_{i1}(\vartheta)
-\frac1n\sum_{i:A_i=0}\xi_i\Oh_{i0}(\vartheta)+\opn{n^{-1/2}}.$
For each arm $j$, substitute the balancing-regular expansion
$\xi_i=(X_i-\Xb_j)^\top\Sh_{X|j}^{-1}(\Xb-\Xb_j)+\rho_i$. The linear term is an
exact identity, and the $\rho_i$ term is $\opn{n^{-1/2}}$ uniformly in $\vartheta$:
\[
\begin{aligned}
\frac1n\sum_{i:A_i=j}\xi_i\Oh_{ij}(\vartheta)
&=(\Xb-\Xb_j)^\top\Sh_{X|j}^{-1}
\frac1n\sum_{i:A_i=j}(X_i-\Xb_j)\Oh_{ij}(\vartheta)
+\frac1n\sum_{i:A_i=j}\rho_i\Oh_{ij}(\vartheta)\\
&=\frac{n_j}{n}(\Xb-\Xb_j)^\top\bh_j(\vartheta)+\opn{n^{-1/2}},
\end{aligned}
\]
uniformly in $\vartheta$, where the last step uses
$\frac1n\sum_{i:A_i=j}(X_i-\Xb_j)\Oh_{ij}(\vartheta)=\frac{n_j}{n}\Sh_{X|j}\bh_j(\vartheta)$
and $\Sh_{X|j}^{-1}\Sh_{X|j}=\mathrm I$. Hence
\[
\Uh_h(\vartheta)-\Uh_\mathrm{unadj}(\vartheta)
=\frac{n_1}{n}(\Xb-\Xb_1)^\top\bh_1(\vartheta)
-\frac{n_0}{n}(\Xb-\Xb_0)^\top\bh_0(\vartheta)+\opn{n^{-1/2}}.
\]
Applying the centering identities
$\frac1n\sum_iA_i(X_i-\Xb)=\frac{n_1}{n}(\Xb_1-\Xb)$ and
$\frac1n\sum_i(1-A_i)(X_i-\Xb)=\frac{n_0}{n}(\Xb_0-\Xb)$, together with
$\Xb-\Xb_j=-(\Xb_j-\Xb)$, yields \eqref{eq:cox_cal_reduction}.
\end{proof}

\subsection{Step 4 -- From score equivalence to root equivalence}
\label{subsec:cox-step5}

Define the running Cox augmentation
\[
A_n(\vartheta)
=
\frac1n\sum_{i=1}^n
\Bigl\{
A_i(X_i-\Xb)^\top\bh_1(\vartheta)
-(1-A_i)(X_i-\Xb)^\top\bh_0(\vartheta)
\Bigr\}.
\]
Then Proposition~\ref{prop:cox_bal_bridge} states
$\Uh_h(\vartheta)=\Uh_\mathrm{unadj}(\vartheta)-A_n(\vartheta)+\opn{n^{-1/2}} $
uniformly over \(\Theta_0\).  Ye's adjusted Cox score instead uses the frozen
augmentation \(A_n(\hat\theta_\mathrm{unadj})\).

\begin{lemma}[Consistency of the weighted Cox root and freezing of the augmentation]
\label{lem:cox_consistency_freeze}
Assume the conditions of Proposition~\ref{prop:cox_bal_bridge}, and let
$g_n(\vartheta)=-\partial_\vartheta\Uh_\mathrm{unadj}(\vartheta)$. Suppose $\Theta_0$ is a
convex neighborhood of $\theta_0$, $\sqrt n\,\Uh_\mathrm{unadj}(\theta_0)=\Opn1$,
$\hat\theta_h\in\Theta_0$ solves $\Uh_h(\vartheta)=0$, and a positive-information
condition holds on $\Theta_0$: $\sup_{\vartheta\in\Theta_0}|g_n(\vartheta)|=\Opn1$,
and $\inf_{\vartheta\in\Theta_0}g_n(\vartheta)\ge g_*$ for some $g_*>0$ with
probability tending to one.
\begin{enumerate}[label=\textup{(\alph*)},leftmargin=2.2em]
\item \emph{($n^{-1/2}$-consistency.)} $\hat\theta_h-\theta_0=\Opn{n^{-1/2}}$.
\item \emph{(Freezing the augmentation.)} If in addition
$\hat\theta_\mathrm{unadj}-\theta_0=\Opn{n^{-1/2}}$, then for any $\vartheta\in\Theta_0$ with
$|\vartheta-\hat\theta_\mathrm{unadj}|=\Opn{n^{-1/2}}$,
\[
A_n(\vartheta)-A_n(\hat\theta_\mathrm{unadj})=\opn{n^{-1/2}},
\qquad
\Uh_h(\vartheta)=\Uh_\mathrm{aug}(\vartheta)+\opn{n^{-1/2}}.
\]
\end{enumerate}
\end{lemma}

\begin{proof}
Throughout, Lemma~\ref{lem:car-mean-balance} gives
$n^{-1}\sum_i A_i(X_i-\Xb)=\tfrac{n_1}{n}(\Xb_1-\Xb)=\Opn{n^{-1/2}}$ and the same
bound with $A_i$ replaced by $1-A_i$, while
Lemma~\ref{lem:cox_Otheta_bounded} gives
$\sup_{\vartheta\in\Theta_0}\|\bh_j(\vartheta)\|=\Opn1$ and
$\sup_{\vartheta\in\Theta_0}\|\partial_\vartheta\bh_j(\vartheta)\|=\Opn1$. In
particular $\sup_{\vartheta\in\Theta_0}|A_n(\vartheta)|=\Opn{n^{-1/2}}$.

\smallskip
\textbf{(a)} Using Proposition~\ref{prop:cox_bal_bridge} at the root
$\hat\theta_h$,
\[
0=\Uh_h(\hat\theta_h)=\Uh_\mathrm{unadj}(\hat\theta_h)-A_n(\hat\theta_h)+\opn{n^{-1/2}}.
\]
By the mean-value theorem, for some $\bar\theta$ between $\hat\theta_h$ and
$\theta_0$, $\Uh_\mathrm{unadj}(\hat\theta_h)=\Uh_\mathrm{unadj}(\theta_0)-g_n(\bar\theta)(\hat\theta_h-\theta_0)$,
and $\bar\theta\in\Theta_0$ by convexity. Hence
\[
g_n(\bar\theta)(\hat\theta_h-\theta_0)
=\Uh_\mathrm{unadj}(\theta_0)-A_n(\hat\theta_h)+\opn{n^{-1/2}}=\Opn{n^{-1/2}}.
\]
The positive-information condition makes $g_n(\bar\theta)$ bounded away from zero
with probability tending to one, so $\hat\theta_h-\theta_0=\Opn{n^{-1/2}}$.

\smallskip
\textbf{(b)} Only the slopes $\bh_j(\cdot)$ in $A_n(\cdot)$ depend on $\vartheta$,
so
\[
|A_n(\vartheta)-A_n(\hat\theta_\mathrm{unadj})|
\le\Big\|\tfrac1n\sum_i A_i(X_i-\Xb)\Big\|\,\|\bh_1(\vartheta)-\bh_1(\hat\theta_\mathrm{unadj})\|
+\Big\|\tfrac1n\sum_i(1-A_i)(X_i-\Xb)\Big\|\,\|\bh_0(\vartheta)-\bh_0(\hat\theta_\mathrm{unadj})\|.
\]
The covariate prefactors are $\Opn{n^{-1/2}}$, and by the mean-value inequality
with $\sup_{u\in\Theta_0}\|\partial_u\bh_j(u)\|=\Opn1$,
\[
\|\bh_j(\vartheta)-\bh_j(\hat\theta_\mathrm{unadj})\|
\le\sup_{u\in\Theta_0}\|\partial_u\bh_j(u)\|\,|\vartheta-\hat\theta_\mathrm{unadj}|=\Opn{n^{-1/2}}.
\]
Therefore $A_n(\vartheta)-A_n(\hat\theta_\mathrm{unadj})=\Opn{n^{-1}}=\opn{n^{-1/2}}$.
Combining with Proposition~\ref{prop:cox_bal_bridge},
$\Uh_h(\vartheta)=\Uh_\mathrm{unadj}(\vartheta)-A_n(\vartheta)+\opn{n^{-1/2}}$, and
$\Uh_\mathrm{aug}(\vartheta)=\Uh_\mathrm{unadj}(\vartheta)-A_n(\hat\theta_\mathrm{unadj})$, gives
$\Uh_h(\vartheta)=\Uh_\mathrm{aug}(\vartheta)+\opn{n^{-1/2}}$.
\end{proof}

\begin{theorem}[Root equivalence for the marginal Cox hazard ratio]
\label{thm:cox_root_equiv}
Assume the conditions of Lemma~\ref{lem:cox_consistency_freeze}, that
$\hat\theta_\mathrm{unadj}-\theta_0=\Opn{n^{-1/2}}$, and that $\hat\theta_\mathrm{aug}\in\Theta_0$
solves $\Uh_\mathrm{aug}(\vartheta)=0$ with $\hat\theta_\mathrm{aug}-\theta_0=\Opn{n^{-1/2}}$.
Then
\[
\sqrt n\,(\hat\theta_h-\hat\theta_\mathrm{aug})\xrightarrow{p}0 .
\]
Under the method-specific conditions of
Propositions~\ref{prop:sbw-car}, \ref{prop:eb-car}, and
\ref{prop:ipw-car}, the conclusion therefore applies to SBW, EB, and stabilized IPW.
\end{theorem}

\begin{proof}
By Lemma~\ref{lem:cox_consistency_freeze}(a),
$\hat\theta_h-\theta_0=\Opn{n^{-1/2}}$. Since
$\hat\theta_\mathrm{unadj}-\theta_0=\Opn{n^{-1/2}}$, Lemma~\ref{lem:cox_consistency_freeze}(b)
applies at $\vartheta=\hat\theta_h$, giving
\[
0=\Uh_h(\hat\theta_h)=\Uh_\mathrm{aug}(\hat\theta_h)+\opn{n^{-1/2}},
\]
so $\Uh_\mathrm{aug}(\hat\theta_h)=\opn{n^{-1/2}}$. Because $\Uh_\mathrm{aug}(\hat\theta_\mathrm{aug})=0$,
\[
\Uh_\mathrm{aug}(\hat\theta_h)-\Uh_\mathrm{aug}(\hat\theta_\mathrm{aug})=\opn{n^{-1/2}}.
\]
The augmentation in $\Uh_\mathrm{aug}$ is frozen at $\hat\theta_\mathrm{unadj}$, so
$\partial_\vartheta\Uh_\mathrm{aug}(\vartheta)=\partial_\vartheta\Uh_\mathrm{unadj}(\vartheta)=-g_n(\vartheta)$.
By the mean-value theorem, for some $\bar\theta$ between $\hat\theta_h$ and
$\hat\theta_\mathrm{aug}$,
\[
\Uh_\mathrm{aug}(\hat\theta_h)-\Uh_\mathrm{aug}(\hat\theta_\mathrm{aug})=-g_n(\bar\theta)(\hat\theta_h-\hat\theta_\mathrm{aug}).
\]
Both roots lie in the convex set $\Theta_0$, so $\bar\theta\in\Theta_0$, and the
positive-information condition of Lemma~\ref{lem:cox_consistency_freeze} makes
$g_n(\bar\theta)$ bounded away from zero with probability tending to one.
Dividing gives $\hat\theta_h-\hat\theta_\mathrm{aug}=\opn{n^{-1/2}}$, i.e.
$\sqrt n\,(\hat\theta_h-\hat\theta_\mathrm{aug})\xrightarrow{p}0$.
\end{proof}

\section{Prognostic covariates, derived outcomes, and the efficiency gain}
\label{supp:prognostic}

This section gives an interpretation of why prognostic baseline covariates lead to efficiency gain in covariate-adjusted log-rank statistic. Let \(q_j(t)=\{1-\mu(t)\}^j\{\mu(t)\}^{1-j}\), where
\(\mu(t)=E\{A_i\mid Y_i(t)=1\}\).  The population derived outcome in the
linearization of the log-rank score is \cite{lin1989robust, ye2024covariate}
\[
O_{ij}
=
\int_0^\tau q_j(t)
\{dN_{ij}(t)-Y_{ij}(t)p(t)\,dt\},
\qquad j=0,1,
\]
where \(p(t)\,dt=E\{dN_i(t)\}/E\{Y_i(t)\}\).  This is the population version of
the derived outcome used by \citet{ye2024covariate}; the arm-specific factor
\(q_j(t)\) is essential. For this interpretive calculation, assume \(C_{ij}\) is independent of
\((T_{ij},X_i)\) within arm \(j\), so that
\(\Pr(C_{ij}\ge t\mid X_i)=G_j(t)\). Since
\(Y_{ij}(t)=1\{T_{ij}\ge t,C_{ij}\ge t\}\),
\[
E\{Y_{ij}(t)\mid X_i\}
=
P(T_{ij}\ge t,C_{ij}\ge t\mid X_i)
=
G_j(t)S_j(t\mid X_i).
\]
Similarly, \(dN_{ij}(t)\) records an observed failure in an infinitesimal
neighborhood of \(t\). Hence
\[
E\{dN_{ij}(t)\mid X_i\}
=
P(T_{ij}\in dt,C_{ij}\ge t\mid X_i)
=
G_j(t)S_j(t\mid X_i)\lambda_j(t\mid X_i)\,dt.
\]
Taking conditional expectation in the definition of \(O_{ij}\) gives
\[
E(O_{ij}\mid X_i)
=
\int_0^\tau q_j(t)G_j(t)S_j(t\mid X_i)
\{\lambda_j(t\mid X_i)-p(t)\}\,dt .
\]
Finally, by the law of total covariance,
\[
\mathrm{Cov}(X_i,O_{ij})
=
\mathrm{Cov}\{X_i,E(O_{ij}\mid X_i)\},
\]
because \(X_i\) is fixed after conditioning on itself. Interchanging covariance
and integration yields
\[
\mathrm{Cov}(X_i,O_{ij})
=
\int_0^\tau q_j(t)G_j(t)
\mathrm{Cov}\!\left[
X_i,\,
S_j(t\mid X_i)\{\lambda_j(t\mid X_i)-p(t)\}
\right]dt .
\]

\subsection{A local Cox prognostic model}

To quantify the connection, suppose \(E(X_i)=0\),
\(\mathrm{var}(X_i)=\Sigma_X\), and, within arm \(j\),
\[
\lambda_j(t\mid X_i)=\lambda_{0j}(t)\exp(\eta^\top X_i),
\]
where \(\eta\) is a prognostic covariate log hazard ratio.  Let
\(\Lambda_{0j}(t)=\int_0^t\lambda_{0j}(u)\,du\) and
\(S_{0j}(t)=\exp\{-\Lambda_{0j}(t)\}\).  A first-order expansion around
\(\eta=0\) gives
\[
S_j(t\mid X_i)
=
S_{0j}(t)\{1-\Lambda_{0j}(t)\eta^\top X_i\}
+O(\|\eta\|^2),
\qquad
\lambda_j(t\mid X_i)
=
\lambda_{0j}(t)\{1+\eta^\top X_i\}
+O(\|\eta\|^2).
\]
Let \(p_0(t)\) denote the pooled hazard at \(\eta=0\).  Since \(p(t)\) is free of
\(X_i\), its own first-order perturbation contributes no covariance with \(X_i\)
at this order.  Therefore
\[
\label{eq:supp_general_local}
\mathrm{Cov}\!\left[
X_i,\,
S_j(t\mid X_i)\{\lambda_j(t\mid X_i)-p(t)\}
\right]
=
S_{0j}(t)
\left[
\lambda_{0j}(t)\{1-\Lambda_{0j}(t)\}
+
p_0(t)\Lambda_{0j}(t)
\right]\Sigma_X\eta
+O(\|\eta\|^2).
\]
This display also shows why a completely general monotonicity statement is
delicate: away from the null, the coefficient in brackets need not be positive
for every \(t\), so cancellation over time is possible.

Under the local null baseline, however,
\(\lambda_{01}(t)=\lambda_{00}(t)=\lambda_0(t)\), so \(p_0(t)=\lambda_0(t)\),
\(S_{01}(t)=S_{00}(t)=S_0(t)\), and the bracket in
\eqref{eq:supp_general_local} simplifies to \(\lambda_0(t)\).  Hence
\[
\mathrm{Cov}(X_i,O_{ij})
=
c_j\,\Sigma_X\eta+O(\|\eta\|^2),
\qquad
c_j
=
\int_0^\tau q_j(t)G_j(t)S_0(t)\lambda_0(t)\,dt>0 .
\]
Consequently,
\[
\beta_j
=
\Sigma_X^{-1}\mathrm{Cov}(X_i,O_{ij})
=
c_j\,\eta+O(\|\eta\|^2),
\qquad j=0,1 .
\]

\cite{ye2024covariate}'s variance reduction has the form
\begin{equation}\label{eq:supp_Delta}
\sigma_L^2-\sigma_{CL}^2
=
\pi(1-\pi)(\beta_1+\beta_0)^\top
\Sigma_X(\beta_1+\beta_0)\ge0,
\end{equation}
with strict inequality unless \(\beta_1+\beta_0=0\).  Substituting
\(\beta_1+\beta_0=(c_1+c_0)\eta+O(\|\eta\|^2)\) into
\eqref{eq:supp_Delta} yields the local approximation
\[
\sigma_L^2-\sigma_{CL}^2
=
\pi(1-\pi)(c_1+c_0)^2
\eta^\top\Sigma_X\eta
+
O(\|\eta\|^3).
\]
Thus, near the null and for sufficiently small nonzero \(\eta\), the efficiency
gain is positive and increases to first order in the quadratic prognostic
strength \(\eta^\top\Sigma_X\eta\).

\section{Data Generation Process Details}
The baseline covariate vector $X_i$ was motivated by components of the Framingham risk score. Define
$L_{\mathrm{Age},i}=\log(\mathrm{Age}_i)$,
$L_{\mathrm{HDL},i}=\log(\mathrm{HDL}_i)$, and
$L_{\mathrm{SBP},i}=\log(\mathrm{SBP}_i)$, and let $\mathrm{DM}_i$ and
$\mathrm{Smoke}_i$ denote the diabetes and current-smoking indicators. Thus,
\[
  X_i=(L_{\mathrm{AGE},i}, \mathrm{DM}_i,L_{\mathrm{HDL},i}, L_{\mathrm{SBP},i}, \mathrm{Smoke}_i)^\top.
\]
The three continuous components were generated jointly as
\[
  \begin{pmatrix}
  L_{\mathrm{Age},i}\\L_{\mathrm{HDL},i}\\L_{\mathrm{SBP},i}
  \end{pmatrix}
  \sim \mathcal{N}_3\!\left(
  \begin{pmatrix}4.0\\3.5\\4.8\end{pmatrix},
  \operatorname{diag}(0.3,0.25,0.1)\,
  \begin{pmatrix}
  1 & 0.245 & 0.239\\
  0.245 & 1 & -0.036\\
  0.239 & -0.036 & 1
  \end{pmatrix}
  \operatorname{diag}(0.3,0.25,0.1)
  \right).
\]
Diabetes and smoking were generated independently as
$\mathrm{DM}_i\sim\operatorname{Bernoulli}(0.20)$ and
$\mathrm{Smoke}_i\sim\operatorname{Bernoulli}(0.15)$. The nonlinear oracle prognostic score was
\[
  Z_i=1-0.9^{\exp(g_i)},
  \qquad
  g_i=2.5L_{\mathrm{Age},i}+0.7\mathrm{DM}_i-0.8L_{\mathrm{HDL},i}
      +2.5L_{\mathrm{SBP},i}+0.5\mathrm{Smoke}_i-19.415 .
\]
This is the absolute risk at a horizon with baseline survival $0.9$ under a
proportional-hazards model with linear predictor $g_i$, the construction used by
Framingham- and PREVENT-type risk scores. Equivalently
$\log\{-\log(1-Z_i)\}=g_i+\log(-\log 0.9)$, so $Z_i$ is the inverse complementary
log-log transform of a linear predictor and $Z_i\in(0,1)$. The intercept
$-19.415$ equals the mean of the linear part, so $g_i$ has mean zero and $Z_i$
has median $1-0.9=0.1$, a 10\% baseline risk.

\paragraph{Randomization designs for scenarios (F) and (G).}
Both stratified scenarios use the binary stratum
$S_i=\mathbf 1(Z_i<0.15)$. In scenario (F), \emph{Stratified Simple}, treatment
is assigned independently within each stratum as
$A_i\mid S_i\sim\operatorname{Bernoulli}(0.5)$. In scenario (G),
\emph{Permuted Block}, participants within each stratum are processed in accrual
order in blocks of size four, and each complete block receives a random
permutation of $(0,0,1,1)$, so that exactly two of every four consecutive
within-stratum subjects are assigned to treatment.

\paragraph{Simulation parameters.}
Table~\ref{tab:supp-simparams} summarizes the fixed parameters of the three
studies.

\begin{table}[htbp]
\centering
\caption{Parameter settings across the three simulation studies.}
\label{tab:supp-simparams}
\begin{tabular}{llll}
\toprule
Study & $\beta_A$ & $\beta_Z$ & MC replicates \\
\midrule
Empirical equivalence & $\log(0.7)$ & $9$ & $10{,}000$ \\
Inference validity     & $0$         & $9$ & $10{,}000$ \\
Prognostic strength    & $\log(0.7)$ & $\{1,\dots,10,12,15\}$ & $8{,}000$ \\
\bottomrule
\end{tabular}
\end{table}

\bibliographystyle{plainnat}
\bibliography{ref}

\end{document}